\documentclass[10pt, conference]{IEEEtran} % Aptara syntax
\IEEEoverridecommandlockouts

\usepackage{cite}
\usepackage{amsmath}
\usepackage{xcolor}
\usepackage{amssymb}
\usepackage{url}
\usepackage{graphicx}
\usepackage{ltl}
\usepackage{booktabs} 
\usepackage{setspace}
\usepackage{todonotes}
\usepackage{wrapfig}
\usepackage{hyperref}
\hypersetup{%
colorlinks=true,           % use colored links
allcolors=blue!70!black,   % use dark blue for all links
pdfstartview=Fit,          % open PDF viewer with side-pane hidden.
breaklinks=true,
}
\usepackage{textcomp}
\usepackage{xcolor}
\def\BibTeX{{\rm B\kern-.05em{\sc i\kern-.025em b}\kern-.08em
    T\kern-.1667em\lower.7ex\hbox{E}\kern-.125emX}}

\usepackage{tikz}
\usetikzlibrary{automata}
\usetikzlibrary{positioning}
\usetikzlibrary{math}
\usetikzlibrary{shapes}
\usetikzlibrary{arrows}
\tikzset{state/.style={circle, draw, minimum size=0.5cm}}

\usepackage{algorithm}
\usepackage{algorithmicx}
\usepackage{algpseudocode}
\newlength{\continueindent}
\usepackage{etoolbox}
\makeatletter
\newcommand*{\ALG@customparshape}{\parshape 2 \leftmargin \linewidth \dimexpr\ALG@tlm+\continueindent\relax \dimexpr\linewidth+\leftmargin-\ALG@tlm-\continueindent\relax}
\apptocmd{\ALG@beginblock}{\ALG@customparshape}{}{\errmessage{failed to patch}}
\makeatother

\usepackage{cleveref}
\newtheorem{assumption}{Assumption}

\crefname{appendix}{Appendix}{Appendix}
\crefname{section}{Section}{Sections}
\crefname{subsection}{Section}{Sections}
\crefname{definition}{Definition}{Definitions}
\crefname{proposition}{Proposition}{Propositions}
\crefname{lemma}{Lemma}{Lemmas}
\crefname{theorem}{Theorem}{Theorems}
\crefname{statement}{Statement}{Statements}
\crefname{corollary}{Corollary}{Corollaries}
\crefname{example}{Example}{Examples}
\crefname{figure}{Figure}{Figures}
\crefname{assumption}{Assumption}{Assumptions}
\crefname{remark}{Remark}{Remarks}
\crefname{table}{Table}{Tables}
\crefname{running}{Running Example}{Running Examples}
\crefname{algorithm}{Algorithm}{Algorithms}

\newcommand{\yw}[1]{#1}

\usepackage{mathtools}
\usepackage{stmaryrd}
\SetSymbolFont{stmry}{bold}{U}{stmry}{m}{n}

\usepackage{xspace}

\renewcommand{\mid}{\,\vert\,}

\renewcommand{\phi}{\varphi}

\newcommand{\pctl}{{\sf \small PCTL}\xspace}
\newcommand{\pctls}{{\sf \small PCTL$^*$}\xspace}
\newcommand{\hpctls}{{\sf \small HyperPCTL$^*$}\xspace}

\newcommand{\hltl}{{\sf \small HyperLTL}\xspace}
\newcommand{\hpctl}{{\sf \small HyperPCTL}\xspace}
\newcommand{\ltl}{{\sf \small LTL}\xspace}

\newcommand{\paths}{\mathsf{Paths}}

\newcommand{\comp}[1]{\textsf{\small #1}}

\newcommand{\nat}{\mathbb{N}}
\newcommand{\real}{\mathbb{R}}
\renewcommand{\subset}{\subseteq}

\renewcommand{\epsilon}{\varepsilon}

\newcommand{\abs}[1]{\vert #1 \vert}
\newcommand{\nm}[1]{\Vert #1 \Vert}

\newcommand{\bin}{\mathrm{Binom}}

\newcommand{\m}{\mathcal{M}}

\newcommand{\ap}{\mathsf{a}}
\newcommand{\aps}{\mathsf{AP}}

\newcommand{\pv}{\pi}
\newcommand{\pvs}{\underline{\pi}}
\newcommand{\sv}{\sigma}

\newcommand{\pvars}{\mathrm{\Pi}}

\newcommand{\ms}{s}

\newcommand{\mss}{\mathcal{S}}

\newcommand{\mts}{\mathbf{T}}
\newcommand{\msi}{\ms_\mathrm{init}}

\newcommand{\ml}{L}

\newtheorem{remark}{Remark}
\newtheorem{lemma}{Lemma}
\newtheorem{theorem}{Theorem}

\newcommand{\pa}{S}

 \newcommand{\G}{\LTLsquare}
 \newcommand{\F}{\LTLdiamond}
 \newcommand{\X}{\LTLcircle}
\newcommand{\U}{\mathbin{\mathcal{U}}}

\renewcommand{\P}{\mathbin{\mathbb{P}}}
\newcommand{\True}{\ensuremath{\mathtt{true}}}

\newcommand{\pr}{\mathbf{Pr}}
\newcommand{\ex}{\mathbb{E}}

\newcommand{\FP}{\alpha_\mathrm{FP}}
\newcommand{\FN}{\alpha_\mathrm{FN}}
\newcommand{\mle}{p^\mathrm{MLE}}

\newcommand{\T}{\mathcal{T}}

\begin{document}

\title{Statistical Model Checking for Hyperproperties}

% \author{Yu Wang,
% Siddhartha Nalluri,
% Borzoo Bonakdarpour,
% and Miroslav Pajic%
% \thanks{Yu Wang, Siddhartha Nalluri, and Miroslav Pajic are with 
% the Department of Electrical \& Computer Engineering, 
% Duke University, Durham, NC 27708, USA. 
% Email: \{yw354, siddhartha.nalluri, miroslav.pajic\}@duke.edu}%
% \thanks{Borzoo Bonakdarpour is with 
% the Department of Computer Science, 
% Iowa State University, Ames, IA 50014, USA. 
% Email: borzoo@iastate.edu}%
% \thanks{This work is sponsored in part by the ONR under agreements 
% N00014-17-1-2504 and N00014-20-1-2745, AFOSR under award number 
% FA9550-19-1-0169, as well as the NSF CNS-1652544 and NSF SaTC-1813388 grant.}
% }

\author{
\IEEEauthorblockN{Yu Wang\IEEEauthorrefmark{1}, 
Siddhartha Nalluri\IEEEauthorrefmark{2}, 
Borzoo Bonakdarpour\IEEEauthorrefmark{3},
and Miroslav Pajic\IEEEauthorrefmark{1}}
\IEEEauthorblockA{\IEEEauthorrefmark{1}Department of Electrical \& Computer Engineering,
Duke University, USA\\
Email: \sf{\{yu.wang094, miroslav.pajic\}@duke.edu}}
\IEEEauthorblockA{\IEEEauthorrefmark{2}Department of Computer Science, Duke University, USA\\
Email: \sf{siddhartha.nalluri@duke.edu}}
\IEEEauthorblockA{\IEEEauthorrefmark{3}Department of Computer Science and Engineering, Michigan State University, USA\\
Email: \sf{borzoo@msu.edu}}
}

\maketitle

\begin{abstract}

\emph{Hyperproperties} have shown to be a powerful tool for expressing and reasoning about information-flow security policies. In this paper, we investigate the problem of {\em statistical model checking} (SMC) for hyperproperties. Unlike exhaustive model checking, SMC works based on drawing samples from the system at hand and evaluate the specification with statistical confidence. The main benefit of applying SMC over exhaustive 
techniques is its efficiency and scalability. 
To reason about probabilistic hyperproperties, we first propose the temporal 
logic \hpctls that extends \pctls and \hpctl.
We show that \hpctls can express important probabilistic information-flow
security policies that cannot be expressed with \hpctl.
Then, we introduce SMC algorithms for verifying \hpctls formulas on 
discrete-time Markov chains, based on sequential probability ratio tests (SPRT) 
with a new notion of multi-dimensional indifference region.
Our SMC algorithms can handle both non-nested and nested
probability operators for any desired significance level.
%
% Then, we extend the proposed SMC
% algorithms to \hpctls specifications with nested probability and temporal
% operators. 
To show the effectiveness of our technique, 
we evaluate our SMC algorithms on four case studies focused on information security: timing side-channel vulnerability in 
encryption, probabilistic anonymity in dining cryptographers, probabilistic 
noninterference of parallel programs, and the performance of a randomized cache 
replacement policy that acts as a countermeasure against cache flush attacks.

\end{abstract}

\section{Introduction} \label{sec:intro}

%\yw{

{\em Randomization} has been a powerful tool in the design and development of 
many algorithms and protocols that make probabilistic guarantees in the area of 
information security. Prominent examples such as {\em quantitative information 
flow}~\cite{kb07,s09}, {\em probabilistic noninterference}~\cite{g92}, and 
{\em differential privacy}~\cite{dr14} quantify the amount of information 
leakage and the relation between 
\yw{two probabilistic execution traces} of a system. These 
and similar requirements constitute {\em probabilistic 
hyperproperties}~\cite{cs10,ab18}. They extend traditional trace properties from 
sets of execution traces to sets of execution traces and allow for 
explicit and simultaneous quantification over the temporal behavior of multiple 
execution traces. Probabilistic hyperproperties stipulate the probability 
relation between independent executions.

{\em Model checking}, an automated technique that verifies the correctness of 
a system with respect to a formal specification, has arguably been the most 
successful story of using formal methods in the past three decades. Since many 
systems have stochastic nature (e.g., randomized distributed algorithms), model 
checking of such systems has been an active area of research. Temporal logics such as \pctls~\cite{baier_PrinciplesModelChecking_2008} as well as model checkers \comp{PRISM}~\cite{knp11} and \comp{STORM}~\cite{DehnertJK017} 
have been developed as formalism and tools to express and reason about probabilistic 
systems. However, these techniques are unable to capture and verify 
probabilistic hyperproperties that are vital to reason about quantified 
information-flow security.

The state of the art in specification and verification of probabilistic 
hyperproperties is limited to the temporal logic \hpctl~\cite{ab18}. The model 
checking algorithm for \hpctl utilizes a numerical approach that iteratively 
computes the exact measure of paths satisfying relevant sub-formulas. In this 
context, we currently face two significant and orthogonal gaps to apply 
verification of probabilistic hyperproperties in~practice:

\begin{itemize}

\item {\bf Expressiveness.} \ First, \hpctl does {\em not} allow (1)~nesting of 
temporal operators, which is necessary to express requirements such as performance guarantees in randomized cache replacement protocols that defend 
against cache-flush attacks, and (2)~explicit quantification over execution paths, which is necessary to reason about the probability of reaching certain~states.  

\item {\bf Scalability.} \ Second, and perhaps more importantly, numerical algorithms for probabilistic model checking, including the one proposed in~\cite{ab18}, tend to require substantial time and space, and often run into 
serious scalability issues. Indeed, these algorithms work only for small systems 
that have certain structural properties. On top of this difficulty, another major challenge in verifying hyperproperties is that the computation 
complexity for exhaustive verification grows at least exponentially in the 
number of quantifiers of the input 
formula~\cite{ab18,ab16,clarkson_TemporalLogicsHyperproperties_2014,bf18}. 

\end{itemize}

In this work, our goal is to address the above stumbling blocks 
(expressiveness and scalability) by investigating {\em statistical model 
checking} (SMC) 
\cite{agha_SurveyStatisticalModel_2018,legay_StatisticalModelChecking_2010,larsen_StatisticalModelChecking_2016}
for hyperproperties with probabilistic guarantees. 
To this end, we first introduce on discrete-time Markov chains 
the temporal logic \hpctls that extends \pctls~\cite{baier_PrinciplesModelChecking_2008} by 
(i) allowing explicit quantification over paths,
and \hpctl~\cite{ab18} 
by (ii) allowing nested probability and temporal operators. 
These two features are crucial in expressing probabilistic hyperproperties,
such as probabilistic noninterference.
Specifically, consider a probabilistic program with a high-security input $h 
\in \{0,1\}$ and a low-security output $l \in \{0,1\}$.
Probabilistic noninterference requires that the probability of observing the low-security output $l=0$ (or $l=1$) should be equal for two executions $\pi_{h=0}$ and $\pi_{h=1}$
that have the high-security input 
$h=0$ and $h=1$, respectively. In other words, the high-security input cannot be 
inferred from the low-security output through a probabilistic channel -- i.e.,
%
% \begin{align*}
% & \P^{\pv_1}( \pi_{h=0} \text{ outputs } l=0 ) 
% = \P^{\pv_2}( \pi_{h=0} \text{ outputs } l=0 ) \, \wedge\\
% & \P^{\pv_1}( \pi_{h=0} \text{ outputs } l=1 ) 
% = \P^{\pv_2}( \pi_{h=0} \text{ outputs } l=1 )
% \end{align*}
$$	
\P^{\pi_{h=0}}( \pi_{h=0} \text{ outputs } l=0 ) 
= \P^{\pi_{h=1}}( \pi_{h=1} \text{ outputs } l=0 )
$$
This property involves the relation between two 
executions $\pi_{h=0}$ and $\pi_{h=1}$,
and cannot be expressed by non-hyper logics, such as \pctls.
We also illustrate that \hpctls can elegantly express properties such as 
generalized probabilistic causation, countermeasures for side-channel attacks, 
probabilistic noninterference, and probabilistic independence among executions.
In addition, the latter is an important performance property for cache 
replacement policies that defend against cache flush attacks and cannot be 
expressed in \hpctl, as it requires using nested temporal operators.

To tackle the scalability problem, we turn to SMC -- a popular 
approach in dealing with probabilistic systems that uses a {\em sample-based} 
technique, where one asserts whether the system satisfies a property by 
observing some of its executions
\cite{legay_StatisticalModelChecking_2010,zuliani_StatisticalModelChecking_2015,
younes_YmerStatisticalModel_2005,roohi_StatisticalVerificationToyota_2017,zarei2020statistical}. 
The general idea of SMC is to treat the problem of checking a temporal logic 
formula on a probabilistic system as {\em hypothesis 
testing}~\cite{sen_StatisticalModelChecking_2005a,
agha_SurveyStatisticalModel_2018}. By drawing samples from the underlying 
probabilistic system, the satisfaction of the formula can be inferred
with high confidence levels. To the best of our knowledge, the work on SMC for 
hyperproperties is limited to~\cite{wzbp19}, where the authors propose 
an SMC algorithm for hyperproperties for cyber-physical 
systems using the {\em Clopper-Pearson} (CP) confidence 
intervals.
\yw{In this work, we propose another SMC algorithm
for hyperproperties using 
{\em sequential probability 
ratio tests} (SPRT)
\cite{wald_SequentialTestsStatistical_1945},
which are more efficient for statistical 
inference than using the confidence 
intervals.}

Developing SMC for \hpctls formulas 
using SPRT has significant challenges 
that do not appear in SMC
for non-hyper probabilistic temporal logics, such as
\pctls.
This is caused by the fact that in \hpctls, one can express complex probabilistic quantification among different paths.
Specifically, \hpctls allows~for:

\begin{itemize}

\item {\bf Probabilistic quantification of multiple paths.} For 
example, formula
\begin{equation} \label{eq:ex1}
\P^{(\pv_1, \pv_2)} (\ap^{\pv_1} \U \ap^{\pv_2}) > p
\end{equation}
means that the probability that an atomic proposition $\ap$ holds on a 
random path $\pv_1$ until it becomes true on another random path $\pv_2$ is 
greater than some $p \in [0,1]$.
    
\item {\bf Arithmetics of probabilistic quantification.} For example, formula
\begin{equation} \label{eq:ex2}
\P^{\pv_1} (\F \ap^{\pv_1}) + \P^{\pv_2} (\G \ap^{\pv_2}) > p  
\end{equation} 
stipulates that the sum of the probability that $\ap$ finally holds 
and the probability that $\ap$ always holds,
is greater than some $p \geq 0$.

\item {\bf Nested probabilistic quantification.} This is different 
from nested probabilistic quantification in \pctls. For example, formula
\begin{equation} \label{eq:ex3}
\P^{\pv_1} \big(\P^{\pv_2} (\ap^{\pv_1} \U \ap^{\pv_2}) > p_1 \big)  > 
p_2,
\end{equation} 
requires that for a (given) path $\pv_1$, the probability that 
$(\ap^{\pv_1} \U \ap^{\pv_2})$ holds for a random path $\pv_2$, is greater than 
some $p_1 \in [0,1]$; and, this fact should hold with probability greater than 
some $p_2 \in [0,1]$ for a random path $\pv_1$.
\end{itemize}
The different kinds of 
complex probabilistic quantification among multiple paths
cannot be handled by existing SMC algorithms 
for non-hyper probabilistic temporal logics~\cite{agha_SurveyStatisticalModel_2018}.

To use SPRT to handle the aforementioned challenges of SMC 
requires a condition on the {\em indifference regions}.
As a simple example, to statistically infer if $\pr (A) > p$, for some random event $A$, 
using SPRT from sampling, it is required that the probability $\pr(A)$ should 
not be too ``close'' to $p$; this means that there exists some known $\varepsilon > 0$ 
such that 
$\pr(A) \notin (p - \varepsilon, p + \varepsilon)$, i.e., $\pr(A) \geq p + \varepsilon$ or $\pr(A) \leq p - \varepsilon$.
This is a common assumption used 
for many SMC techniques~\cite{sen_StatisticalModelChecking_2005a,
agha_SurveyStatisticalModel_2018}.
Therefore, it is sufficient to test between the two most indistinguishable cases
$\pr(A) \notin p - \varepsilon$
and
$\pr(A) \notin p + \varepsilon$.
The interval $(p - \varepsilon, p + \varepsilon)$ 
is usually referred to as the \emph{indifference~region}.
\yw{In this work, we propose new conditions 
on the {\em indifference regions} 
that enable the use of SPRT in the SMC of \hpctls.}
%
% Because of this assumption, SPRT  is generally much more sample-efficient than 
% confidence interval-based methods.
% %
% This is especially beneficial for SMC of formulas in \hpctls, as the number 
% of samples needed to verify such formulas will grow significantly with the 
% number of random paths in the~formula.

For the SMC of arithmetics of 
probabilistic quantifications in~\eqref{eq:ex2}, 
we consider the hypothesis testing problem:
\begin{equation} \label{eq:into_ht}
	\begin{split}
		& H_0: \big( \P^{\pv_1} (\F \ap^{\pv_1}), \P^{\pv_2} (\G 
\ap^{\pv_2}) \big) \in D,
		\\ & H_1: \big( \P^{\pv_1} (\F \ap^{\pv_1}), \P^{\pv_2} (\G 
\ap^{\pv_2}) \big) \in D^\mathrm{c},
	\end{split}
\end{equation}
where $D = \{(p_1, p_2) \in [0,1]^2 \mid p_1 + p_2 > p \}$ 
and $D^\mathrm{c}$ is its complement set.
To handle the joint probability $\big( \P^{\pv_1} (\F \ap^{\pv_1}), 
\P^{\pv_2} (\G \ap^{\pv_2}) \big)$ in \eqref{eq:into_ht},
we propose a novel {\em multi-dimensional} extension of the standard SPRT. 
%
%To our knowledge, this is the first work to extend SPRT to multiple events.
Specifically, we first generalize the notion of the indifference 
region (namely, the parameter $\varepsilon$) to a multi-dimensional case.
This new notion of indifference region ensures that 
our multi-dimensional SPRT algorithm provides provable probabilistic guarantees
for any desired false positive $\FP \in (0,1)$ and false negative $\FN \in 
(0,1)$ ratios.
\yw{Then we note that the hypotheses $H_0$ and $H_1$ 
in \eqref{eq:into_ht} are composite,
which contains infinitely many simple hypotheses.
To use SPRT, which mainly deal with simple hypotheses, 
on the two composite hypotheses, we propose a geometric condition 
to identify the two most indistinguishable simple hypotheses from
$H_0$ and $H_1$, respectively.
We show that if the SPRT can distinguish these two simple hypotheses,
then any two simple hypotheses from $H_0$ and $H_1$ can be distinguished
by the same test.}

For the SMC of probabilistic quantification of multiple paths
in~\eqref{eq:ex1},
we note that the SMC of 
probabilistic quantification of multiple parallel paths can be handled by 
generalizing the common SPRT to tuples of samples. 
For the SMC of nested probabilistic quantification
in~\eqref{eq:ex3},
we can perform a compositional analysis 
for the probabilistic error in the SMC of the sub-formulas, to yield the global 
false positive and false negative ratios, in the same way as~\cite{wzbp19}.
%
%So, they will not be discussed extensively in this paper.

Finally, based on the above new statistical inference algorithms,
we design SMC algorithms for \hpctls.
These algorithms are fully implemented and evaluated by four prominent case studies.%
\footnote{The simulation code is available 
at~\cite{gitlab_hpctls}.}
Specifically, we apply our SMC algorithms to analyze:
(i)~the time side-channel vulnerability in encryption~\cite{chen2017precise,tizpaz2018data,GabFeed}, (ii)~probabilistic anonymity in dining cryptographers~\cite{chaum_DiningCryptographersProblem_1988}, (iii)~probabilistic noninterference of parallel programs~\cite{goguen1982security},
and (iv)~the performance of a random cache replacement 
policy~\cite{canones2017security} that defends against cache flush attacks.
Our results show that the proposed SMC algorithms 
provide the correct answer with high confidence levels 
in all cases while requiring very short analysis times.

\paragraph{Organization} The rest of the paper is organized as follows.
We introduce \hpctls 
in~\cref{sec:hpctls}.
The expressiveness of \hpctls is 
discussed in \cref{sec:express}, before illustrating its application 
in~\cref{sec:applications}. Our SMC algorithms~for \hpctls are introduced 
in~\cref{sec:non-nested}. We present our case studies and experimental 
results in~\cref{sec:simulation}. Related work is discussed 
in~\cref{sec:related}, before concluding remarks 
in~Section~\ref{sec:conclusion}.
% All proofs are available in the complete 
% manuscript in {\sf \small arXiv} (\url{https://arxiv.org/abs/1902.04111}).

% \input{prelim}
%TEX root=main.tex

\section{The Temporal Logic HyperPCTL$^*$}
\label{sec:hpctls}

We begin with some notation. We denote the set of natural and real 
numbers 
by $\nat$ and $\real$, respectively.
Let $\nat_\infty = \nat \, \cup \, \{\infty\}$.
For $n \in \nat$, let $[n] = \{1,\dots,n\}$.
%
% The {\em indicator} function is denoted by~$\id$.
%
The cardinality of a set is denoted by $\abs{\cdot}$.
For $n \in \nat$, we use $\underline{s} = (s_1, \ldots, s_n)$
%\todo{M: can this be an infinite tuple? if not, why $\inat$} 
to denote a {\em tuple}. We use $\pa = s(0) s(1) \dots$ to denote a {\em sequence}, and 
the $i$-suffix of the sequence is denoted by $\pa^{(i)} = s(i) s(i+1)\cdots$.
For any set $D \subset \real^n$, we denote its {\em boundary}, {\em interior}, {\em closure}
and {\em complement} by $\partial D$, $D^\circ$, $\overline{D}$, and $D^\mathrm{c}$, respectively.

Our proposed temporal logic \hpctls is an extension of \pctls~\cite{baier_PrinciplesModelChecking_2008} that enables 
handling hyperproperties.
It also can be viewed as a variation of \hpctl~\cite{ab18} that 
allows for nested temporal and probability operators.
In this section, we introduce the formal syntax and semantics of
\hpctls; its relation with
\pctls, \hltl and \hpctl is discussed in the next section.

\subsection{Syntax} \label{sub:syntax}

\hpctls formulas are defined by the grammar
\begin{align}
    & \varphi \Coloneqq \;
    \ap^\pv 
    \ \vert \ \varphi^\pv 
    \ \vert \ \neg \varphi
    \ \vert \ \varphi \land \varphi
    \ \vert \ \X \varphi
    \ \vert \ \varphi \U^{\leq k} \varphi
    \ \vert \ \rho \Join \rho
    % \ \vert \ \exists \sv^{\pvs}. \, \varphi
    % \ \vert \ \forall \sv^{\pvs}. \, \varphi \label{eq:hpctls_1}
    \\ & \rho \Coloneqq  \ f(\rho, \dots, \rho)
     \ \vert \ \P^{\pvs} (\varphi)
     \ \vert \ \P^{\pvs} (\rho) \label{eq:hpctls_2}
\end{align}
where

\begin{itemize}
\item $\ap \in \aps$ is an atomic proposition;

\item $\pv$ is a (fresh) {\em random path 
variable} from an infinite supply of such variables $\pvars$;%
\footnote{Technically, using a non-fresh path variable can be allowed.
However, to avoid possible confusion 
about the meaning of the \hpctls formulas, 
we only use fresh path variables here.}

% \item {in Rule 2 of~\eqref{eq:hpctls_1}, $\pv$ should be a fresh random path variable that has not appeared in $\varphi$ to prevent formulas like $(\ap^{\pv_1})^{\pv_1}$ or $((\ap^{\pv_1})^{\pv_2})^{\pv_1}$.}

\item $\X$ and $\U^{\leq k}$ are the `next' and `until' operators, 
respectively, where $k \in \nat_\infty$ is the time bound and $\U^{\leq \infty}$ 
means \emph{``unbounded until}";

\item $\Join \ \in \{<,>,=,\leq,\geq\}$, which allows comparing probabilities among different random paths;

% \item $\sv$ is a {\em state variable} from an infinite supply of state variables $\svs$, and $\forall \sv^{\pvs}$ and $\exists \sv^{\pvs}$ 
% stand for the {\em universal} and {\em existential} quantification of the 
% initial state $\sv$ of a tuple of random path variables $\pvs = 
% (\pv_1, \ldots, \pv_n)$, respectively;

\item $\P^{\pvs}$ is the probability operator for a tuple of random path variables $\pvs = (\pv_1, \ldots, \pv_n)$ for some $n \in \nat$, and

\item $f: \real^n \rightarrow \real$ is an $n$-ary elementary function,%
\footnote{Elementary functions are defined as a sum, product, and/or composition of finitely many polynomials, rational functions, trigonometric and exponential functions, and their inverse functions.}
with constants being viewed as a $0$-ary function. This enables expressing arithmetic operations and entropy from probabilities.
\end{itemize}

\hpctls can be viewed as a probabilistic adaptation of 
\hltl~\cite{clarkson_TemporalLogicsHyperproperties_2014}.
Following the terminology of \hltl 
(and more generally, 
the first-order logic~\cite{enderton2001mathematical}), 
in a given \hpctls formula, we call a path variable \emph{free} 
if it has not been associated by a probability operator; 
otherwise, the path variable is \emph{quantified}.
%
% This is  also.
%
For example, in a \hpctls formula 
$\P^{\pv_1} (\ap^{\pv_1} \U \ap^{\pv_2})$,
the path variable $\pv_1$ is quantified
and the path variable $\pv_2$ is free.
Mostly, we are interested in \hpctls formulas
with all the path variables quantified.

Additional logic operators are derived as usual: $\True \equiv 
\ap^{\pv} \vee \neg \ap^{\pv}$,
$\varphi \lor \varphi' \equiv \neg (\neg \varphi \land \neg \varphi')$,
$\varphi \Rightarrow \varphi' \equiv \neg \varphi \lor \varphi'$, $\F^{\leq k} 
\varphi
\equiv \True \U^{\leq k} \varphi$, 
and $\G^{\leq k} \varphi \equiv \neg \F^{\leq k} \neg \varphi$. 
We denote $\U^{\leq \infty}$, $\F^{\leq \infty}$, and $\G^{\leq 
\infty}$ by $\U$, $\F$, and $\G$, respectively.
We represent a $1$-tuple by its element,
i.e., $\sv^{(\pv)}$ and $\P^{(\pv)}$ are written as $\sv^\pv$ and $\P^\pv$,
respectively.

\subsection{Semantics}
\label{sub:semantics}

We consider the semantics of \hpctls on discrete-time Markov chains (DTMCs) with their states labeled by a set of atomic propositions $\aps$.
Formally, a DTMC is a tuple $\m = (\mss, \msi, \mts, \aps, \ml)$ where 

\begin{itemize}
\item $\mss$ is the finite set of {\em states}, and $\msi$ the {\em initial 
state};

\item $\mts: \mss \times \mss \rightarrow [0, 1]$ is the {\em transition
probability function}, where for any state $\ms \in \mss$, it holds that
$$
\sum_{\ms' \in \mss} \mts(\ms, \ms') = 1;
$$

\item $\aps$ is the set of {\em atomic propositions}, and

\item $\ml : \mss \rightarrow 2^\aps$ is a {\em labeling function}. 
\end{itemize}
An example DTMC labeled by the atomic propositions $\{\ap_1, \allowbreak \ap_2\}$ is 
illustrated in \cref{fig:thm2}.
A {\em path} of a DTMC $\m = (\mss, \msi, \mts, \aps, \ml)$ is of the form $\pa 
= \ms(0) \ms(1) \cdots$, such that for every $i \in \nat$, 
%
% \begin{itemize}
%
% \item $\ms(i) \in \mss$, and 
%
% \item $\mts(\ms(i),\ms({i+1})) \neq 0$. 
%
% \end{itemize}
%
$\ms(i) \in \mss$ and $\mts(\ms(i),\ms({i+1})) \neq 0$. 
By $\paths(\ms)$, we denote the set of paths that start from state $\ms$, while 
$\paths(\m)$ denotes the set of all paths of DTMC~$\m$.

The semantics of \hpctls formulas is described in terms of the interpretation
tuple $(\m, V)$, where

\begin{itemize}

\item $\m = (\mss, \msi, \mts, \aps, \ml)$ is a DTMC, and

\item $V: \pvars \rightarrow \paths(\m)$ is a \emph{path assignment}, mapping 
each (random) path variable to a concrete path of $\m$, starting from the 
initial state $\msi$ by default.
\end{itemize}
We denote by $\llbracket \cdot \rrbracket_{V}$ the 
instantiation of the assignments $V$ on a \hpctls formula.
The judgment rules for semantics of a \hpctls formula $\varphi$ are detailed in~\cref{fig:semantics}, where

\begin{figure*}[!t]
\centering
\fbox{
\begin{minipage}{.956\textwidth}

\[
\begin{array}{l@{\hspace{2.2em}}c@{\hspace{2.2em}}l}
(\m, V) \models \ap^\pv & \textrm{iff} & \ap \in \ml \big( V(\pv) (0) \big)\\
(\m, V) \models \varphi^\pv & \textrm{iff} & (\m, V[\pv' \mapsto V(\pv) 
% \textrm{ for all } \pv' \neq \pv
]) \models \varphi
\\
(\m, V) \models \neg \varphi & \textrm{iff} & (\m, V) \not\models \varphi
\\
(\m, V) \models \varphi_1 \land \varphi_2 &  \textrm{iff} & (\m, V) 
\models
\varphi_1 \textrm{ and } (\m, V) \models \varphi_2
\\
(\m, V) \models \X \varphi &  \textrm{iff} & (\m, V^{(1)}) \models \varphi
\\
(\m, V) \models \varphi_1  \U^{\leq k}  \varphi_2  & \textrm{iff} &
\textrm{there exists } i \leq k \textrm{ such that } \big((\m, V^{(i)}) \models 
\varphi_2 \big) \, \land \,
 \big(\text{for all } j < i.\ (\m, V^{(j)}) \models \varphi_1 \big)\\
(\m, V) \models \rho \Join \rho & \textrm{iff} &
(\m, V) \models \llbracket \rho \rrbracket_{V} \Join \llbracket \rho 
\rrbracket_{V}
\\
 \llbracket f(\rho, \ldots, \rho) \rrbracket_{V} & \textrm{=} & 
 f \big(\llbracket \rho \rrbracket_{V}, \ldots, \llbracket \rho 
\rrbracket_{V} \big)
\\
 \llbracket \P^{(\pv_1, \ldots, \pv_n)} (\varphi) \rrbracket_{V} & 
\textrm{=} &
\pr\Big\{ \big( \pa_i \in \paths(V (\pi_i) (0) \big)_{i \in [n]} :  (\m, 
V[\pv_i \mapsto \pa_i \textrm{ for all } i \in 
[n]]) \models \varphi \Big\}
\\
 \llbracket \P^{(\pv_1, \ldots, \pv_n)} (\rho) \rrbracket_{V} & \textrm{=} 
&
\pr \Big\{ \big( \pa_i \in \paths(V (\pi_i) (0) \big)_{i \in [n]} : 
 (\m, V[\pv_i \mapsto \pa_i \textrm{ for all } i \in [n]]) 
\models \rho \Big\} 
\end{array}
\]

\end{minipage}
}
\caption{Semantics of \hpctls.}
\label{fig:semantics}
\end{figure*}

\begin{enumerate}
  \item $V[\cdot]$ denotes the revision of the assignment 
  $V$ by the rules given in $[\cdot]$.

  \item $V^{(i)}$ is the $i$-shift of path assignment $V$, defined by  $V^{(i)}(\pv) = (V(\pv))^{(i)}$.
  
  \item By the second rule, associating the path variable 
$\pv$ to the formula $\varphi$ 
assigns the value of all path variables in $\varphi$ to $V(\pv)$.
For a given $(\m, V)$, 
the satisfaction of $\varphi$ is preserved, if the 
free path variables in $\varphi$ are replaced by $\pv$.
(The quantified path variables are unaffected, as discussed in Point 4) below.)

For example, the following formulas are 
\emph{semantically} equivalent
-- i.e., the truth value of the formulas on both sides are identical
for any given $(\m, V)$,
\[
\begin{split}
  (\ap^{\pv_1})^{\pv_2} & \equiv \ap^{\pv_2},
\\  (\X \ap^{\pv_1})^{\pv_2} & \equiv \X \ap^{\pv_2},
\\  (\ap^{\pv_1} \U \ap^{\pv_2})^{\pv_3} & \equiv 
  \ap^{\pv_3} \U \ap^{\pv_3}.
\end{split}  
\]
In particular, in the first above equivalence, $\pv_2$ in  
$(\ap^{\pv_1})^{\pv_2}$ can replace $\pi_1$, since $\pv_1$ is a free random 
path variable and obtain $\ap^{\pv_2}$. However, these two formulas would not 
be equivalent if $\pv_1$ was not free.

  \item In the last two rules, the probability $\pr$ is taken 
for an $n$-tuple of sample paths $(\pa_1, \ldots, \pa_n)$ to instantiate 
$\pvs$, and ``$:$'' means `such that'.
The evaluation of the probability operator $\llbracket \P^{(\pv_1, \ldots, \pv_n)} (\varphi) \rrbracket_{V}$
means to (re)draw the random path variables $\pv_1, \ldots, \pv_n$ from their current initial states on the DTMC $\m$
(regardless of their current assignment by $V$), 
and evaluate the satisfaction probability of $\varphi$.
Thus, following Point 3), the quantified path variables 
are unaffected by the association of new path variables
and we have the following \emph{semantic} equivalence:
\[
\begin{split}
  \big( \P^{\pv_1} (\ap^{\pv_1} \U \ap^{\pv_2}) 
  \big)^{\pv_3} & \equiv 
  \P^{\pv_1} (\ap^{\pv_1} \U \ap^{\pv_3}),
  \\
  \big( \P^{(\pv_1, \pv_2)} (\ap^{\pv_1} \U \ap^{\pv_2}) 
  \big)^{\pv_3} & \equiv 
  \P^{(\pv_1, \pv_2)} (\ap^{\pv_1} \U \ap^{\pv_2}).
\end{split}  
\]
\end{enumerate}
In particular, in the first equivalence, $\pv_3$ in $(\P^{\pv_1} 
(\ap^{\pv_1} \U \ap^{\pv_2}))^{\pv_3}$ can replace $\pv_2$, since $\pv_2$ is 
free, obtaining $\P^{\pv_1} (\ap^{\pv_1} \U \ap^{\pv_3})$. However, $\pv_3$ 
cannot replace $\pv_1$, as $\pv_1$ is quantified by the probability operator.

% \todo{In the second rule, why there is index 0? (prob left over of state 
% quant?) The last rule is not corret I 
% think. we don't have satisfation relation for $\rho$. We have it for [[]]. 
% Rules 8 doesn't need $\m \models$. They need equality.}
%

\subsection{Discussion on HyperPCTL$^*$}

Consider the DTMC $\m$ in \cref{fig:thm2}, and the following \hpctls formula:
\begin{equation*} \label{eq:ex_1}
\varphi = \P^{(\pv_1, \pv_2)} \Big((\ap_1^{\pv_1} \land \ap_1^{\pv_2}) \; 
\land \; \F (\ap_2^{\pv_1} \land \ap_2^{\pv_2}) \Big) > 1/6.
\end{equation*}
The formula claims that two (independently) random paths $\pv_1$ and 
$\pv_2$ from $\msi = s_0$ satisfy $(\ap_1^{\pv_1} \land \ap_1^{\pv_2}) \; \land 
\; \F (\ap_2^{\pv_1} \land \ap_2^{\pv_2})$, 
i.e., both paths should satisfy $\ap_1$ in their initial state 
and satisfy $\ap_2$ 
(later) at the same time
with probability greater than $1/6$.
\yw{By calculation from \cref{fig:thm2}, 
this probability is $1/4$, so we have $\m \models \varphi$.}

\begin{figure}[t]
\centering
\begin{tikzpicture}
  \node[state, initial, initial text=, label={$\{\ap_1\}$}] (0) 
{$\ms_0$};
  \node[state,right of=0, label={$\{\}$},xshift=.5cm] (1) {$\ms_1$};
  \node[state,right of=1, label={$\{\}$},xshift=.5cm, yshift=1cm] (2) {$\ms_2$};
  \node[state,right of=1, label={$\{\ap_2\}$},xshift=.5cm, yshift=-1cm] (3)
{$\ms_3$};

  \path[->] (0) edge node[above] {$1$} (1)
  (1) edge node[above] {$\frac{1}{2}$} (2)
  (1) edge node[below] {$\frac{1}{2}$} (3)
  (2) edge[loop right] node[right] {$1$} ()
  (3) edge[loop right] node[right] {$1$} ();
\end{tikzpicture}
\caption{\hpctls example on DTMC~$\m$.}
\label{fig:thm2}
\end{figure}

\hpctls can generate complex nested formulas.
We explain this using two formulas. First consider the formula:
\begin{equation}
\label{eq:ex4}
  \P^{\pv_1} \Big( \F \big( \P^{(\pv_2, \pv_3)} \big( \ap^{\pv_2} \U (\ap^{\pv_3})^{\pv_1} \big) > c_2 \big) 
\Big) > c_1.
\end{equation}
The formula \eqref{eq:ex4} states that with probability greater than $c_1$, we 
can find a path $\pv_1$, such that finally from some state $\ms$ on $\pv_1$, 
with probability greater than $c_2$, we can find a pair of paths 
$(\pv_2,\pv_3)$ from the pair of states $(\msi, \ms)$ to satisfy 
``$\ap^{\pv_2}$ until $\ap^{\pv_3}$''. %\todo{this is not fully OK}
That is, the computation tree of $\pv_3$ is a subtree of the computation tree 
of $\pv_1$ (rooted at $\msi$), since $\pi_3$ in $(\ap^{\pi_3})^{\pv_1}$ is in 
the scope of $\pv_1$. On the other hand, since $\pv_2$ is indexed by 
$\pv_1$, its computation tree is rooted at $\msi$ (see \cref{fig:ex4}).
The inner subformula $\P^{(\pv_2, \pv_3)} \big( \ap^{\pv_2} \U 
(\ap^{\pv_3})^{\pv_1} \big) > c_2$ 
in \eqref{eq:ex4} involves the probabilistic computation trees of 
$\pv_2$ and $\pv_3$, as shown by the dotted box in \cref{fig:ex4}.

\begin{figure}[t]
  \centering
  \begin{tikzpicture}
    \draw[fill] (0, 0) ellipse (2pt and 2pt);
    \draw[->, >=stealth] (0, 0) node[left] {$\msi$} -- (1, 0);
    \draw[dashed, ->, >=stealth] (0, 0) -- node[above] 
    {$\pv_1$} (0.6, 0.3);
    \draw[dashed, ->, >=stealth] (0, 0) -- (0.6, -0.3);
    \draw[fill] (1, 0) ellipse (2pt and 2pt);
    \draw[->, >=stealth] (1, 0) node[below] {$\ms$} -- (2, 0);
    \draw[dashed, ->, >=stealth] (1, 0) -- node[above] {$\pv_3$} (1.6, 0.3);
    \draw[dashed, ->, >=stealth] (1, 0) -- (1.6, -0.3);
    \draw[fill] (1, -1) ellipse (2pt and 2pt);
    \draw[->, >=stealth] (1, -1) node[below] {$\msi$} -- (2, -1);
    \draw[dashed, ->, >=stealth] (1, -1) -- node[above] {$\pv_2$} (1.6, -0.7);
    \draw[dashed, ->, >=stealth] (1, -1) -- (1.6, -1.3);
    \draw[dotted] (0.7, 0.7) rectangle (2.2, -1.7);
  \end{tikzpicture}

  \caption{Computation trees for~\eqref{eq:ex4}. The dashed arrows show 
other possible sample values of the path variables to illustrate the 
probabilistic computation~tree. }
  \label{fig:ex4}
\end{figure}
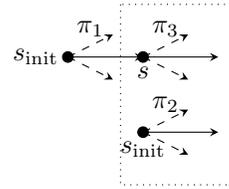

Now, consider the formula:
\begin{equation}
\label{eq:ex5}
  \P^{\pv_1} \Big( \F \big( \P^{(\pv_2, \pv_3)} \big( \ap^{\pv_2} \U \ap^{\pv_3} 
\big) > c_2 \big)^{\pv_1} 
\Big) > c_1.
\end{equation}
It requires that with probability greater than $c_1$, we 
can find a path $\pv_1$, such that finally from some state $\ms$ on $\pv_1$, 
with probability greater than $c_2$, we can find a pair of paths $(\pv_2, 
\pv_3)$ from the state $\ms$ that satisfy ``$\ap^{\pv_2}$ until $\ap^{\pv_3}$''.
That is, the computation tree of $\pv_2$ and $\pv_3$ (rooted at $\ms$) is a 
subtree of the computation tree of $\pv_1$, rooted at $\msi$ (see 
\cref{fig:ex5}).
Again, the inner subformula 
$\P^{(\pv_2, \pv_3)} \big( \ap^{\pv_2} \U \ap^{\pv_3} 
\big) > c_2$ 
in \eqref{eq:ex5} involves the probabilistic computation trees of 
$\pv_2$ and $\pv_3$, as shown by the dotted box in \cref{fig:ex5}.
%\todo{shouldn't this be (8), and the fight formula from (8)? or then point to \cref{fig:ex4}; if former, it is somewhat repetitive}

\begin{figure}[t]
  \centering
  \begin{tikzpicture}
    \draw[fill] (0, 0) ellipse (2pt and 2pt);
    \draw[->, >=stealth] (0, 0) node[left] {$\msi$} -- (1, 0);
    \draw[dashed, ->, >=stealth] (0, 0) -- node[above] 
    {$\pv_1$} (0.6, 0.3);
    \draw[dashed, ->, >=stealth] (0, 0) -- (0.6, -0.3);
    \draw[fill] (1, 0) ellipse (2pt and 2pt);
    % \draw[->, >=stealth] (1, 0) node[below] {$\ms$} -- node[above] {$\pv_3$} 
(2, 0);
    \draw[->, >=stealth] (1, 0) -- (1.6, 0.6) node[right] {$\pv_2$};
    \draw[dashed, ->, >=stealth] (1, 0) -- (1.6, 0.3);
    \draw[dashed, ->, >=stealth] (1, 0) -- (1.3, 0.6);
    \draw[->, >=stealth] (1, 0) -- (1.6, -0.6) node[right] {$\pv_3$};
    \draw[dashed, ->, >=stealth] (1, 0) -- (1.6, -0.3);
    \draw[dashed, ->, >=stealth] (1, 0) -- (1.3, -0.6);
    \draw[dotted] (0.7, 1) rectangle (2.2, -1);
    %   \draw[fill] (1, -1) ellipse (2pt and 2pt);
  %   \draw[->, >=stealth] (1, -1) node[below] {$\msi$} -- node[above] 
  % {$\pv_2$} (2, -1);
  %   \draw[dashed, ->, >=stealth] (1, -1) -- (1.6, -1.6);
  %   \draw[dotted] (0.7, 0.7) rectangle (2.2, -1.7);
  \end{tikzpicture}

  \caption{Computation trees for~\eqref{eq:ex5}. The dashed arrows show 
other possible sample values of the path variables to illustrate the 
probabilistic computation~tree.}
  \label{fig:ex5}
\end{figure}
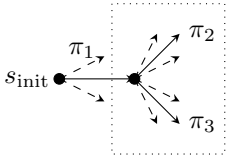

% Finally, we note that although 
% $(\ap^{\pv_3})^{\pv_1}$ and $\ap^{\pv_1}$
% are semantically equivalent as mentioned above,
% the semantic equivalence does not hold if 
% we replace $(\ap^{\pv_3})^{\pv_1}$ with $\ap^{\pv_1}$ 
% in \eqref{eq:ex4}.

\section{Applications of HyperPCTL$^*$}
\label{sec:applications}

In this section, we illustrate the application of \hpctls by four examples related to information-flow security, ranging from timing attacks, scheduling of parallel programs, communication protocols, and computer hardware.
\yw{These examples cannot be properly handled by existing temporal logics.}

\subsection{Side-channel Vulnerability}
\label{sub:sc_vul}

Timing side-channel attacks are possible if an attacker can infer the secret 
values, \yw{which are set at the second step of an execution}, by observing the execution time of a program. 
To prevent such attacks, it is required that the 
probability of termination within some $k \in \nat$ steps should be 
approximately equal for
two (random) executions $\pv_1$ and $\pv_2$, 
where the secret values are 
$\mathtt{S}_1$ and $\mathtt{S}_2$, respectively:
\begin{equation} \label{eq:sc_vul}
\begin{split}
	& \P^{\pv_1} \big( (\X \mathtt{S}_1^{\pv_1}) \Rightarrow 
(\F^{\leq k} \mathtt{F}^{\pv_1}) \big) 
	\\ & \qquad \approx_\varepsilon 
	\P^{\pv_2} \big( (\X \mathtt{S}_2^{\pv_2}) \Rightarrow 
(\F^{\leq k} \mathtt{F}^{\pv_2}) \big),
\end{split}
\end{equation}
where the label $\mathtt{F}$ represents the end of execution, the 
next operator signifies that the secret is established in the first step of 
execution, and
$\approx_\varepsilon$ stands for approximately equal within some $\varepsilon > 
0$.
If~\eqref{eq:sc_vul} holds, then an attacker cannot infer the the secret values 
from whether the program terminates in $k$~steps. %\todo{doesn't this make it sufficient?}

\subsection{Probabilistic Noninterference}
\label{sub:noninterference}

Probabilistic noninterference~\cite{g92} establishes the connection between 
information theory and information flow by employing probabilities to address 
covert channels. Intuitively, it requires that the probability of every 
low-observable trace pattern is the same for every low-equivalent initial state. 
For example, consider the parallel composition of the following $n$-threads:
\begin{equation}
\begin{split}
\mathrm{Th}_k: \texttt{ for }  i_k = 1 & \texttt{ to } (h+1) \times k 
\\ & \texttt{ do }
\{\ldots; l \gets (k \textrm{ mod } 2)\},
\end{split}
\end{equation}
where $k \in [n]$
and $l \in \{0,1\}$ is a publicly observable output.
The secret input $h$ is randomly set to $0$ or $1$ with probability $0.5$.
At each step, the processor randomly chooses one thread among the unfinished threads with equal probability and executes \emph{one iteration} of the for-loop (including the assignment of $l$), until all the $n$ threads are finished.
Clearly, the (random) execution of this $n$-thread program
can be represented by a DTMC, where the states are labeled by the values of all the variables.
Starting from the initial state, it sets the value of $h$ at the second step
and then executes the threads until finished.
The termination states are labeled by $\mathtt{F}$.

%
% The DTMC describing the change of the values of the variables 
% $(i_1, \ldots, i_k, \ldots, i_n, h, l)$ is
% \begin{equation}
% \begin{split}
% & \Pr \Big( 
% 	(i_1, \ldots, i_k+1, \ldots, i_n, h, k \textrm{ mod } 2) 
% 	\, \vert \, (i_1, \ldots, i_k,
% \\ & \qquad \ldots, i_n, h, l) \Big) = \frac{1}{\abs{\{k \in [n] \, \vert i_k \, < (h + 1) \times k\}}},
% \end{split}
% \end{equation}
% where $\abs{\cdot}$ stands for the cardinality.
% %
% Initially, $i_1, \ldots, i_k+1, \ldots, i_n$ are $1$, 
% except for $h$ being either $0$ or $1$.

As the threads have different numbers of loops depending on $h$ and the scheduling is uniformly random, the whole process is more likely to terminate at a thread with more loops, whose thread number is partially indicated by $l$.
This opens up the possibility that by observing $l$, an attacker can infer the difference in the number of loops among the threads, and hence infer $h$.
On the other hand, the attack cannot happen if the probability of observing 
$\mathtt{L_0}: l=0$ (or $\mathtt{L_1}: l=1$) is approximately equal,
regardless of $\mathtt{H_0}: h=0$ or $\mathtt{H_1}: h=1$
-- i.e., the value of $h$ cannot be inferred from the value of $l$.
This is formally defined in \hpctls by:
\begin{equation} \label{eq:def_pn}
\begin{split}
& \P^{\pv_1} \Big( (\X \mathtt{H_0}^{\pv_1}) \Rightarrow 
\big( \F ( \mathtt{F}^{\pv_1} \land \mathtt{L_0}^{\pv_1}) \big) \Big) 
\\ & \quad \approx_\varepsilon \P^{\pv_2} \Big( (\X \mathtt{H_1}^{\pv_2}) \Rightarrow 
\big( \F ( \mathtt{F}^{\pv_2} \land \mathtt{L_0}^{\pv_2}) \big) \Big),
\end{split}
\end{equation}
and 
\begin{equation} \label{eq:def_pn'}
\begin{split}
& \P^{\pv_1} \Big( (\X \mathtt{H_0}^{\pv_1}) \Rightarrow 
\big( \F ( \mathtt{F}^{\pv_1} \land \mathtt{L_1}^{\pv_1}) \big) \Big) 
\\ & \quad \approx_\varepsilon \P^{\pv_2} \Big( (\X \mathtt{H_1}^{\pv_2}) \Rightarrow 
\big( \F ( \mathtt{F}^{\pv_2} \land \mathtt{L_1}^{\pv_2}) \big) \Big),
\end{split}
\end{equation}
where $\approx_\varepsilon$ stands for approximately equal within $\varepsilon$ 
and the next operator signifies that the secret is established in the first 
step of execution.
In~\eqref{eq:def_pn}, $\pv_1$ is a random execution of the program, where it 
sets $h = 0$ at the second step and finally yields $l=0$ and $\pv_2$ is a 
random execution of the program, where it sets $h = 1$ at the second step and 
finally yields $l=0$; and similarly for~\eqref{eq:def_pn'}.

\subsection{Dining Cryptographers}
\label{ssub:dc}

Several cryptographers sit around
a table having dinner. Either one of the cryptographers or, alternatively,
the National Security Agency (NSA) must pay for their meal. The cryptographers
respect each other's right to make an anonymous payment but want to find out
whether the NSA paid. So they decide to execute the following~protocol:
\begin{itemize}

\item Every two cryptographers establish a shared one-bit secret by tossing an unbiased coin and only informs the cryptographer on the right of the outcome.

\item Then, each cryptographer publicly states whether the two coins that it can see (the one it flipped and the one the left-hand neighbor flipped) agree if he/she did not~pay.

\item However, if a cryptographer actually paid for dinner, then it instead states the opposite  -- disagree if the coins are the same and agree if the coins are different.

\item An even number of agrees indicates that the NSA paid, while an odd number indicates that a cryptographer~paid.

\end{itemize}

The protocol can be modeled by a DTMC with the states labeled by the values of the Boolean variables mentioned below.
In addition, the state labels $\mathtt{C_i}$ for $i = 1,2,3$ indicate that
cryptographer $i$ paid, and $\mathtt{C_0}$ indicates that the NSA paid.
The common shared secret between two cryptographers $i$ and $j$ is indicated by the label $\mathtt{S}_{ij}$.
The final result of the process is indicated by a Boolean variable $\mathtt{P}$, where $\mathtt{P}$ if a cryptographer paid, and $\neg \mathtt{P}$ otherwise.
We define an information-flow security condition that given that some cryptographer paid, the probability that either cryptographer $i$ or $j$ paid are (approximately) equal irrespective of the common shared secret between them, i.e., the results of the coin tosses.
This is specified by the following \hpctls formula:
\begin{equation} \label{eq:def_dc}
\begin{split}
	%   \forall \sv_1^{\pv_1}. \forall \sv_2^{\pv_2}. & \forall 
% % \sv_3^{\pv_3}. \forall \sv_4^{\pv_4}. \ 
% 	 (\mathtt{C_i}^{\pv_1} \land \mathtt{C_i}^{\pv_2} 
% 	 \land \mathtt{C_j}^{\pv_3} \land  \mathtt{C_j}^{\pv_4} ) \;  \Rightarrow 
& \P^{\pv_1} \big( \F ( \neg \mathtt{S}_{ij}^{\pv_1} \land \F 
\mathtt{P}^{\pv_1} ) \big) 
	 \approx_\varepsilon \P^{\pv_2} \big( \F
( \mathtt{S}_{ij}^{\pv_2} \land \F \mathtt{P}^{\pv_2} ) \big)       
 \\ & \quad \approx_\varepsilon
\P^{\pv_3}  \big( \F
( \neg \mathtt{S}_{ij}^{\pv_3} \land \F \mathtt{P}^{\pv_3} ) \big) 
\approx_\varepsilon 
\P^{\pv_4} \big( \F
( \mathtt{S}_{ij}^{\pv_4} \land \F \mathtt{P}^{\pv_4} ) \big).
\end{split}
\end{equation}
where $\approx_\varepsilon$ stands for approximately equal within $\varepsilon$.
In~\eqref{eq:def_dc}, $\pv_1$ is a random execution of the protocol, 
where the common shared secret between two cryptographers $i$ and $j$ is set to $\mathtt{S}_{ij}$
during the execution and the final return is $\mathtt{P}$ -- i.e. some cryptographer paid;
and similarly for $\pv_2$, $\pv_3$, and $\pv_4$.

\subsection{Randomized Cache Replacement Policy} \label{sub:cache}

Cache replacement policies decide which \emph{cache lines} are replaced in case of a \emph{cache 
miss}. Randomized policies employ random replacement as a countermeasure against 
cache flush attacks. On the negative side, they also 
introduce performance losses.
Following~\cite{canones2017security}, we model a cache as a \emph{Mealy 
machine} with the access sequence as the input.
Each state of the Mealy machine represents a unique configuration of the cache,
i.e., the cache lines stored.
The transition of the Mealy machine captures a \emph{random replacement policy} 
that for access to memory data in address $b$,
(i) if it is already stored in the cache, return Hit $\mathtt{H}$;
(ii) if it is not stored and the cache has free space, return Miss 
$\mathtt{M}$ and write $b$ in free space, and
(iii) if $b$ is not stored and the cache is full,  
then returns Miss $\mathtt{H}$, and randomly overwrite a line (with uniform distribution) with $b$.

The performance requirement of such a policy is that, from an empty cache, 
after $N$ steps (when the cache almost fills), in a time window of $T$, the 
probability of observing 
$T$ consecutive $\mathtt{H}$ should be greater than that of 
observing $\mathtt{H}$ only $T-1$ times in that window.
This is formally expressed as:
\begin{align} \label{eq:def_rcr}
% \forall \sv^{(\pv_1, \pv_2)}. \; & 
% (\mathtt{B}^{\pv_1} \land \mathtt{B}^{\pv_2}) \Rightarrow \big( 
\P^{\pv_1}( \X^{(N)} \G^{\leq T} \mathtt{H}^{\pv_1} ) >  
\P^{\pv_2}( \X^{(N)} \varphi^{\pv_2} ) + \varepsilon,
% \big).
\end{align}
where $\varepsilon > 0$ is a parameter, $\varphi^{\pv_2}$ means there is one $\texttt{M}$ for $N$ consecutive accesses, formally expressed as
\[
\begin{split}
& 
\varphi^{\pv_2} = \big( \texttt{M}^{\pv_2} \land \X \texttt{H}^{\pv_2} \land  \ldots \land 
\X^{(T-1)} \texttt{H}^{\pv_2} \big)
\\ & \qquad
\lor \ldots \lor \big( \texttt{H}^{\pv_2} \land \ldots \land \X^{(T-2)} \texttt{H}^{\pv_2} \land  \X^{(T-1)} \texttt{M}^{\pv_2} \big)
\end{split}
\]
where $\mathtt{B}$ indicates the initial state of an empty cache, and $\X^{(N)}$ represents the $N$-fold composition of $\X$.
In~\eqref{eq:def_rcr}, $\pv_1$ is a random execution of the cache replacement policy, 
where starting from the step $N$, there are $T$ consecutive hits $\texttt{H}$;
$\pv_2$ is a random execution, 
where starting from the step $N$, there is only one miss $\texttt{M}$ for the next $T$ steps.

\subsection{Generalized Probabilistic Causation} 
\label{sub:probabilistic_causation}

\hpctls can express conditional probabilities over multiple independent 
computation trees, which is not possible in 
\hpctl~\cite{ab18}.
{\em Probabilistic causation}~\cite{hitchcock_ProbabilisticCausation_2018} 
asserts that
if the cause~$\psi^{\pvs}$ happens,
the probability of occurring an effect~$\varphi^{\pvs}$ should be higher than 
the probability of occurring 
$\varphi^{\pvs}$ when $\psi^{\pvs}$ does not happen.
\yw{Here, we allow the cause and effect 
to be hyperproperties to capture probabilistic causality 
between security properties, e.g., the existence of 
a side-channel (see \cref{sub:sc_vul}) 
results in another side-channel.}
We can specify that for any two premises 
(i.e., initial states), $\psi^{\pvs}$ probabilistically causes $\varphi^{\pvs}$ 
as follows:
\begin{equation} \label{eq:def_pc}
	\frac{\P^{\pvs_1} (\psi^{\pvs_1} \land \varphi^{\pvs_1})}{\P^{\pvs_2} 
(\psi^{\pvs_2})} 
	> \frac{\P^{\pvs_3} (\neg \psi^{\pvs_3} \land 
\varphi^{\pvs_3})}{\P^{\pvs_4} (\neg \psi^{\pvs_4})}.
\end{equation}
In \eqref{eq:def_pc}, $\pvs_1$ is a tuple of random executions
where both the cause $\psi$ and effect $\varphi$ hold;
$\pvs_2$ a tuple of random executions 
where the cause $\psi$ holds.
Thus, the left-hand side of \eqref{eq:def_pc}
is the conditional probability of the effect, when the cause holds.
Similarly, the right-hand side of \eqref{eq:def_pc}
is the conditional probability of the effect, when the cause 
does not hold.

The cause and effect in \eqref{eq:def_pc} can themselves be hyperproperties.
For instance, the cause can be the violation of probabilistic noninterference 
(i.e.,~\eqref{eq:def_pn}) and the effect can be a breach of safety. 
That is, 
leakage of information increases the probability of compromising safety.
This probabilistic causation of hyperproperties cannot
be expressed by \pctls or 
any of its existing extensions 
including~\cite{ar08}.

\section{Relation to Other Temporal Logics}
\label{sec:express}

In this section, we illustrate the expressive power of \hpctls
by comparing it with 
\pctls~\cite{baier_PrinciplesModelChecking_2008},
\hpctl~\cite{ab18},
and \hltl~\cite{clarkson_TemporalLogicsHyperproperties_2014}.

\subsection{Relation to PCTL$^*$} 
\label{sec:hpctls-pctls}

In a \pctls formula, a probability operator {\em implicitly} incorporates a 
single random sample path drawn from a (probabilistic) computation 
tree.
In \hpctls, such random path variables are explicitly specified.
For example, checking the nested \pctls formula
$$
\P^{J_1} \big( \X (\P^{J_2} (\varphi)) \big),
$$
involves two random sample paths from a root computation tree (for 
$\P^{J_1}$) and a sub computation tree (from the second state of the first 
path for $\P^{J_2}$), respectively.
Thus, in order to specify this formula in \hpctls, we need to explicitly employ 
two random path variables $\pv_1$ and $\pv_2$ for the two probability 
operators, where sub-formula $\varphi$ is checked on $\pv_2$ of the sub 
computation tree, whose root is randomly given by $\pv_1(1)$ 
(see \cref{fig:ex0}).
\begin{figure}[t]
\centering
  \begin{tikzpicture}
    \draw[fill] (0, 0) ellipse (2pt and 2pt);
    \draw[->, >=stealth] (0, 0) node[below] {$\ms$} -- (1, 0);
    \draw[dashed, ->, >=stealth] (0, 0) -- node[above] {$\pv_1$} (0.6, 0.3);
    \draw[dashed, ->, >=stealth] (0, 0) -- (0.6, -0.3);
    \draw[fill] (1, 0) ellipse (2pt and 2pt);
    \draw[->, >=stealth] (1, 0) node[below] {$\pv_1(1)$} -- (2, 0);
    \draw[dashed, ->, >=stealth] (1, 0) -- node[above] 
    {$\pv_2$} (1.6, 0.3);
    \draw[dashed, ->, >=stealth] (1, 0) -- (1.6, -0.3);
  \end{tikzpicture}
  \caption{Computation trees for~\eqref{eq:ex0}. The dashed arrows show other possible sample values of the path variables to illustrate the probabilistic 
computation~tree.}
  \label{fig:ex0}
\end{figure}
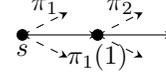
Hence, sub-formula $\X (\P^{J_2} (\varphi))$ is checked on $\pv_1$. The 
corresponding \hpctls formula is: 
\begin{equation} \label{eq:ex0}
\P^{\pv_1} \big( \X (\P^{\pv_2} 
(\varphi^{\pv_2}) \in J_2)^{\pv_1} \big) \in J_1.
\end{equation}

Formally, we first show that \hpctls subsumes \pctls.
This is done by providing the set of rules 
to translate every \pctls formula to a \hpctls formula.
We use the syntax and semantics of the \pctls from~\cref{sub:app_pctls}.

\begin{theorem}
\label{thrm:hpcts-pctls}
\hpctls subsumes \pctls.
\end{theorem}

\begin{proof}
We prove this statement by showing that 
any \pctls formula can be transformed into a \hpctls formula 
with the same meaning.
In other words, for any given DTMC, the satisfaction/dissatisfaction 
of the formula is preserved during the transformation.

Given a DTMC $\m$,
the satisfaction 
of a \pctls state formula $\Phi$
(as defined in \cref{sub:app_pctls})
transforms into the satisfaction of a \hpctls formula by
\begin{equation} \label{eq:transformation}
  \m \models \Phi \textrm{ if and only if } (\m, V) \models \T(\Phi, \pv),
\end{equation}
for any path assignment $V$.
In \eqref{eq:transformation}, 
the \pctls state formula $\Phi$ implicitly involves a random path
(more precisely, a random computation tree),
which is explicitly named by $\pi$ in 
the corresponding \hpctls formula.
The transformation $\T$ is defined inductively as follows: %\todo{let's talk about $\pi$}}
\begin{itemize}
  \item $\T(\ap, \pv) = \ap^\pv$

  \item $\T(\neg \varphi, \pv) = \neg \T(\varphi, \pv)$ 
  
  \item $\T(\neg \Phi, \pv) = \neg \T(\neg \Phi, \pv)$
  
  \item $\T(\varphi_1 \land \varphi_2, \pv) = \T(\varphi_1, \pv) \land 
  \T(\varphi_2, \pv)$
  
  \item $\T(\Phi_1 \land \Phi_2, \pv) = \T(\Phi_1, \pv) \land \T(\Phi_2, \pv)$ %\todo{is $\pi$ missing on the left side?}
  
  \item $\T(\X \varphi, \pv) = \X \T(\varphi, \pv)$
  
  \item $\T(\varphi_1 \U^{\leq k} \varphi_2, \pv) = \T(\varphi_1, \pv) \U^{\leq 
  k} \T(\varphi_2, \pv)$

  \item $\T(\P^{J} (\varphi), \pv) = \big( \P^{\pv'} (\T(\varphi, \pv')) \in J 
\big)^\pv$
  with $\pv' \neq \pv$,
\end{itemize}
where $\Phi$ is a \pctls state formula, $\varphi$ is a \pctls path formula 
  (as defined in \cref{sub:app_pctls}).
The correctness of the transformation follows directly from the semantics of the logics.
The transformation~\eqref{eq:transformation} holds 
for any path assignment $V$, 
since it can be shown that the path variables in $\T(\Phi, \pv)$ 
are all (probabilistically) quantified 
and actually do not receive assignment from $V$.
\end{proof}

Next, we show that \hpctls \emph{strictly} subsumes \pctls.
Specifically, we construct a DTMC and a \hpctls formula,
and show that this formula cannot be expressed by~\pctls.

\begin{theorem}
\label{stat:hpcts-pctls_strict}
\hpctls is strictly more expressive than \pctls with respect to DTMCs.
\end{theorem}

\begin{proof}
Consider the DTMC shown in \cref{fig:thm1} and the following \hpctls 
formula:
\[
\varphi = \bigg(\frac{\P^{\pv_1}
\big(\mathsf{init}^{\pv_1} \Rightarrow \F (\ap_1^{\pv_1} \land
\ap_2^{\pv_1})\big)}{\P^{\pv_2} \big(\mathsf{init}^{\pv_2} \Rightarrow \F
\ap_2^{\pv_2} \big)} = \frac{1}{2}\bigg).
\]
Now, we prove that $\varphi$ cannot be expressed in \pctls. By the syntax and 
semantics of \pctls, it suffices to show that $\varphi$ cannot be expressed by 
a formula $\P (\psi)$, where $\psi$ is a \pctls path formula derived by 
concatenating  a set of \pctls state formulas $\Phi_1, \ldots, \Phi_n$ with 
$\land, \neg$, or the temporal operators. These state formulas are either true 
or false in the states $\ms_0$, $\ms_1$, $\ms_2$, and $\ms_3$.
Thus, the satisfaction of $\psi$ defines a subset of the paths $\paths(\ms_0) 
= 
\{ \ms_0 \ms_1^{\omega}, \ms_0 \ms_2^{\omega}, \ms_0 \ms_3^{\omega}\}$ 
in the DTMC. 
Since every path in $\paths(\ms_0)$ is taken with probability $1/3$, 
formula $\P (\psi)$ can only evaluate to a value in $\{0, 1/3, 2/3, 1\}$.
However, by the semantics of \hpctls, the fractional probability on the right 
side of the implication has value $1/2$; thus, $\varphi$ evaluates to true and 
cannot be expressed by $\P (\psi)$~in~\pctls.
\end{proof}

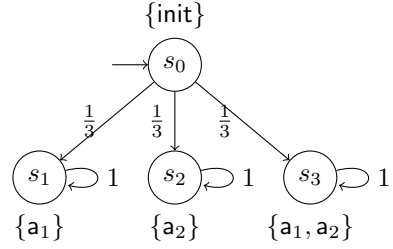
\begin{figure}[t]
    \centering
    \begin{tikzpicture}
      \node[state, initial, initial text=, label={$\{\mathsf{init}\}$}] (0) 
{$\ms_0$};
      \node[state, below of=0, label=below:{$\{\ap_1\}$},
    xshift=-1.8cm,yshift=-.5cm] (1) {$\ms_1$};
      \node[state, below of=0, label=below:{$\{\ap_2\}$},yshift=-.5cm] (2)
    {$\ms_2$};
      \node[state, below of=0, label=below:{$\{\ap_1,\ap_2\}$},
    xshift=1.8cm,yshift=-.5cm] (3) {$\ms_3$};
    
      \path[->] (0) edge node[left] {$\frac{1}{3}$} (1)
      (0) edge node[left] {$\frac{1}{3}$} (2)
      (0) edge node[left] {$\frac{1}{3}$} (3)
      (1) edge[loop right] node[right] {$1$} ()
      (2) edge[loop right] node[right] {$1$} ()
      (3) edge[loop right] node[right] {$1$} ();
    \end{tikzpicture}
    \caption{DTMC where \hpctls strictly subsumes \pctls.\label{fig:thm1}}
\end{figure}

\subsection{Relation to HyperPCTL}

Similar to \hpctl~\cite{ab18}, \hpctls allows probability arithmetics and comparison.
For example, the \hpctls formula
$$\varphi = \big( \P^{\pv_1} (\F \ap^{\pv_1}) - \P^{\pv_2} (\F \ap^{\pv_2}) > c 
\big)$$
for some $c \in \real$ means the satisfaction probability of ``finally $\ap$'' 
is greater at least by $c$ on a random path variable $\pv_1$ than another 
random path variable $\pv_2$.
But in general, \hpctl and \hpctls do not subsume each other.
% Unlike \hpctl, \hpctls does not allow for the existential and universal 
% quantification of the initial states, although such quantification can be 
% added to \hpctls. On the contrary, \hpctl does not allow explicit path 
% quantification and nested temporal operators, where \hpctls does.

\begin{theorem}
\label{thm:hpctls-hpctl}

On DTMCs, \hpctls strictly subsumes \hpctl.

\end{theorem}

\begin{proof}
From \cref{stat:hpcts-pctls_strict}, \hpctls subsumes \pctls.
However, \hpctl does not subsume \pctls~\cite{ab18}. 
Thus, \hpctl does not subsume \hpctls.

More specifically, \hpctl cannot express the satisfaction probability of a 
formula with more than two nested temporal operators.
For example, the \hpctls formula $\P^{(\pv_1, \pv_2, \pv_3)} 
(\ap_1^{\pv_1} \U (\ap_2^{\pv_2} \U \ap_3^{\pv_3}))$
cannot be expressed by \hpctl.
This is similar to the fact that \pctl cannot express the satisfaction 
probability of an \ltl formula with more than two nested temporal operators
(but \pctls can).
%
% Another minor difference is that 
% \hpctl does not allow elementary functions of probabilities, 
% such as $\log (\cdot)$.
%

On the other hand, \hpctls contains all the syntactic rules of \hpctl, except 
for the (existential and universal) state quantifications~\cite{ab18}.
A \hpctl formula with state quantifications can be expressed by \hpctls by 
enumerating over the finite set of states of the DTMC.
For example, \hpctl can specify that ``there exists a state $\ms$, 
such that $\P^{(\pv_1, \pv_2)} (\ap^{\pv_1} \U \ap^{\pv_2}) > p$,
where the initial state of $\pv_1, \pv_2$ is $\ms$."
To express this in \hpctls, we introduce an extra
initial state that goes to all the states $\mss$ of DTMC  
with probability $1 / \abs{\mss}$.
Then, the \hpctl specification can be expressed by
$$
\bigvee_{\ms \in \mss} \P^{(\pv_1, \pv_2)} 
\Big( \X \big( (\ms^{\pv_1} \land \ms^{\pv_2}) 
\land (\ap^{\pv_1} \U \ap^{\pv_2}) \big) \Big) > p/\abs{\mss},
$$
where ``$\X$" appears because the paths $\pv_1, \pv_2$ now start from the new initial state.
\end{proof}

\subsection{Relation to HyperLTL}
%
% HyperLTL is an extension of LTL which allows explicit and simultaneous 
% quantification over traces (see Appendix~\ref{sec:hltl})
%
%\paragraph*{Probabilistic path quantification of HyperLTL formulas:}
A \hltl formula can have multiple path variables. For example, let 
$$
\varphi_{\mathsf{hltl}} = \ap_1^{\pv_1} \U \ap_2^{\pv_2}
$$
be a \hltl subformula (i.e., without path quantification), meaning that $\ap_1$ 
is true on $\pv_1$ until $\ap_2$ is true on $\pv_2$.
Like \pctls, which allows for reasoning over the satisfaction probability of 
\ltl formulas, \hpctls allows for reasoning over the satisfaction probability 
of 
\hltl formulas.
For example, \hpctls subformula
%\begin{equation}
%\label{eq:ex1}
$\P^{(\pv_1, \pv_2)} (\varphi_{\mathsf{hltl}}) > c$ %\end{equation}
means that the satisfaction probability of the \hltl formula 
$\varphi_{\mathsf{hltl}}$ is 
greater than $c$.
% , where the initial states of $\pv_1$ and $\pv_2$ are given by 
% $X$ and may be different. 
%
% \paragraph*{Different ways of probabilistic path quantification for HyperLTL
% formulas:}
Moreover, in \hpctls, a \hltl formula can be probabilistically quantified in 
multiple ways.
Specifically, the path variables of the \hltl formula can be quantified at 
one time, or one-by-one in a certain order.
For example, instead of quantifying the \hltl formula $\varphi$ in a 
one-shot way for $\varphi_{\mathsf{hltl}}$, \hpctls also allows formula
$$
\psi_1 = \P^{\pv_1} \big( \P^{\pv_2} (\varphi_{\mathsf{hltl}}) > c_2 \big) > 
c_1.
$$
This means that 
the probability for finding path $\pv_1$ should be 
greater than $c_1$, such that the probability for finding another path $\pv_2$ 
to satisfy $\varphi_{\mathsf{hltl}}$ is greater than $c_2$.
% , where the initial 
% states of $\pv_1$ and 
% $\pv_2$ are given by $X$ and may be different. 
%
By flipping the order of the probabilistic quantification for $\pv_1$ and 
$\pv_2$, we derive the formula 
$$
\psi_2 = \P^{\pv_2} \big( \P^{\pv_1} (\varphi_{\mathsf{hltl}}) > c_2 \big) > 
c_1.
$$
Clearly, the meaning of $\psi_1$ and $\psi_2$ is different, 
showing the significance of the order of the probabilistic quantification.

\section{Statistical Model Checking}
\label{sec:non-nested}

In this section, we design statistical model checking (SMC) algorithms for 
\hpctls formulas on labeled discrete-time Markov chains.
As with previous works on 
SMC~\cite{sen_StatisticalModelChecking_2004,agha_SurveyStatisticalModel_2018,
legay_StatisticalModelChecking_2010},
we focus on handling probabilistic operators 
by sampling.
The temporal operators can be handled in the same way
as for \hltl~\cite{clarkson_TemporalLogicsHyperproperties_2014},
and thus will not be discussed here.

\subsection{Challenges in Developing SMC for HyperPCTL$^*$}

To statistically verify \hpctls, the main challenge is to
use {\em sequential probability ratio tests} (SPRT) to handle 
the following issues:
\begin{itemize}
\item \textbf{Probabilistic quantification of multiple paths.} Consider the 
following formula:
\begin{equation} \label{eq:smc1}
\P^{\pvs} (\varphi) > p, 
\end{equation}
where $\pvs = (\pv_1, \ldots, \pv_n)$ is a tuple of path variables.
Unlike the conventional SMC techniques, evaluating such a formula requires 
drawing {\em multiple} samples (we assume the truth value of $\varphi$ can be 
determined, given the sample value for $\pvs$).

\item \textbf{Arithmetics of probabilistic quantifications.} \ Consider the 
following formula:
\[
f \big( \P^{\pvs_1} (\varphi_1), \ldots,  \P^{\pvs_n} (\varphi_n) \big) 
> p,
\]
where for $i \in [n]$, $\pvs_i$ is a tuple of path variables and the truth 
value of $\varphi_i$ can be determined, given the sample value for $\pvs_i$.
Equivalently, this can be expressed as 
\begin{equation} \label{eq:smc2}
\big( \P^{\pvs_1} (\varphi_1), \ldots,  \P^{\pvs_n} (\varphi_n) \big)  
\in D,
\end{equation} 
where
\[
D = \{ (x_1, \ldots, x_n) \in [0, 1]^{n} \mid f (x_1, \ldots, x_n) > p 
\}.
\]
This can be viewed as an application of the currying technique
in first-order logic that builds the equivalence between functions and 
relations~\cite{enderton2001mathematical}.
In addition, since the functions $f$ is elementary from the syntax of \hpctls, 
the boundary of the domain $D$ is also elementary.

\item \textbf{Nested probabilistic quantification.} \ Consider the following 
formula:
\begin{equation} \label{eq:smc3}
\P^{\pvs_1} \P^{\pvs_2} \cdots \P^{\pvs_n} (\varphi) > p, 
\end{equation} 
where $i \in [n]$, $\pvs_i$ is a tuple of path variables and the truth value of 
$\varphi$ can be determined, given the sample value for all $\pvs_i$. This 
type of formula poses a challenge since the multiple paths drawn for each 
probability operator can be different from its previous or next~operator.
\end{itemize}
These probabilistic quantifications are unique to \hpctls, therefore, they are 
not directly supported by existing statistical model checking algorithms 
designed for non-hyper probabilistic temporal 
logics~\cite{agha_SurveyStatisticalModel_2018}. In the next subsections, we 
address these challenges.

\subsection{Probabilistic Quantification of Multiple Parallel Paths}
\label{sub:3a}

Consider the formula~\eqref{eq:smc1} again. 
We denote the satisfaction probability of the subformula $\varphi$ in~\eqref{eq:smc1} for a given DTMC $\m$ and path assignment $V$ by:
\begin{equation} \label{eq:3a_finite}
	\begin{split}
		p_\varphi = & \pr \Big\{ \Big( \pa_i \in \paths \big( V (\pi_i) (0) 
	\big) \Big)_{i \in [n]} \ : 
		\\ & \qquad 
		\Big( \m, V \big[ \pv_i \mapsto \pa_i \textrm{ for all } i \in [n] \big] \Big) 
	   \models \varphi \Big\}.
	\end{split}
\end{equation}
Following the standard procedure
\cite{agha_SurveyStatisticalModel_2018,larsen_StatisticalModelChecking_2016},  
to simplify our discussion, we first assume that $\varphi$ is a bounded-time specification, i.e., its truth value can be evaluated on the finite 
prefixes of the sample paths.
Unbounded-time specifications can be handled similarly
with extra considerations on the time horizon.
In addition, we make the following assumption on 
the \emph{indifference region}.

\begin{assumption} \label{ass:sprt_indiff}
The satisfaction probability of $\varphi$ is not within 
the indifference region 
$(p - \varepsilon, p + \varepsilon)$ for some $\varepsilon > 0$; 
i.e.,
\begin{equation} \label{eq:indiff_a}
	p_\varphi \notin (p - \varepsilon, p + \varepsilon).
\end{equation}
\end{assumption}

From \cref{ass:sprt_indiff}, to statistically verify~\eqref{eq:smc1},
it suffices to solve the following hypothesis testing (HT) problem:
\begin{equation} \label{eq:ht}
H_0: p_\varphi \leq p - \varepsilon, 
\quad H_1: p_\varphi \geq p + \varepsilon.
\end{equation}
The \yw{hypotheses $H_0$ and $H_1$ in \eqref{eq:ht} are \emph{composite}
since each of them contains infinitely many \emph{simple} hypotheses of
the form $H_0: p_\varphi = p_0$ and $H_1: p_\varphi = p_1$, respectively,}
where $p_0 \in [0, p - \varepsilon]$ and $p_1 \in [p + \varepsilon, 1]$.

\yw{To handle composite hypotheses with SPRT,
a common technique is to consider}
the two most ``indistinguishable'' simple hypotheses
\begin{equation} \label{eq:ht_smc1}
H_0: p_\varphi = p - \varepsilon, 
\quad H_1: p_\varphi = p + \varepsilon
\end{equation}
from the two composite hypotheses in \eqref{eq:ht}, respectively.
From~\cite{sen_StatisticalModelChecking_2004}, if existing samples can test 
between $p - \varepsilon$ and $p + \varepsilon$ for some given statistical 
errors, then these samples are sufficient to test between 
$p - \varepsilon$ and $p_\varphi$ with
the true satisfaction probability
$p_\varphi \in [p + \varepsilon, 1]$
(or between $p + \varepsilon$ and $p_\varphi$ with
the true satisfaction probability
$p_\varphi \in [0, p - \varepsilon]$)
for the same statistical errors
(see \cref{sub:app_smc} for details).

\begin{remark} \label{rem:indifference}
The indifference region assumption
is necessary. If $\varepsilon = 0$,
then $H_0$ and $H_1$ in~\eqref{eq:ht_smc1} 
will be identical.
\end{remark}

To statistically test between $H_0$ and $H_1$
from~\eqref{eq:ht_smc1},
suppose we have drawn $N$ 
statistically independent
sample path tuples
$\underline{\pa}_1, \ldots, \underline{\pa}_N$
for the path variable $\pvs$ from the DTMC.
Let $T$ be the number of sample path tuples,
for which $\varphi$ is true.
This is similar to the statistical model checking of 
\pctls (see \cref{sub:app_smc} for detailed description), 
except that the truth value of $\varphi$ needs to be evaluated for tuples of paths 
instead of single paths.
Let us define, for $x \in (0, 1)$, the log-likelihood function~as 
\begin{equation} \label{eq:lambda_x}
    \lambda(x) = \ln \big( x^{T} (1 - x)^{N - T} \big),	
\end{equation}
then $\lambda(p-\varepsilon)$ and $\lambda(p + \varepsilon)$ 
are the log-likelihood of the two hypotheses 
$H_0$ and $H_1$ in~\eqref{eq:ht_smc1}, respectively.
As the number of sample path tuples $N$ increase,
the log-likelihood ratio
$\lambda(p+\varepsilon) - \lambda(p - \varepsilon)$
should increase (with high probability) if $H_1$ holds,
and should decrease if $H_0$ holds.
To achieve desired the
\emph{false positive} (FP) and \emph{false negative} (FN) ratios 
$\FP$ and $\FN$, respectively, defined by\footnote{Here, $\pr (\cdot \mid \cdot)$ stands for 
the conditional probability.}:
\begin{equation} \label{eq:FPFN}
\begin{split}
& \FP = \pr \big( \text{assert } H_1 \mid  H_0 \textrm{ is true} \big), 
\\ & \FN = \pr \big( \text{assert } H_0 \mid  H_1 \textrm{ is true} \big),
\end{split}
\end{equation}
the SPRT algorithm should continue sampling,
i.e., increase the number of samples $N$,
until one of the two following 
termination conditions hold
\cite{wald_SequentialTestsStatistical_1945}:

\begin{equation} \label{eq:sprt}
\begin{cases}
\text{assert } H_0, & \text{ if } \lambda(p-\varepsilon) - \lambda(p+\varepsilon) > \ln \frac{1 - \FN}{\FP}, \\
\text{assert } H_1, & \text{ if } \lambda(p+\varepsilon) - \lambda(p-\varepsilon) > \ln \frac{1 - \FP}{\FN}.
\end{cases}
\end{equation}
This process is summarized by \cref{alg:sprt}.

\begin{algorithm}[!t]
\caption{SMC of $\P^{\pvs} (\varphi) > p$.}
\label{alg:sprt}
\begin{algorithmic}[1]
\Require Desired FP and FN ratios $\FP$ and $\FN$, indifference parameter 
$\varepsilon$.

\State $N \gets 0$, $T \gets 0$.

\While{True}

\State $N \gets N + 1$.

\State Draw a tuple of sample paths $\underline{\pa}_N$ (from the DTMC).

\If{$\varphi$ is true on $\underline{\pa}_N$}

\State $T \gets T + 1$.

\EndIf

\State Update $\lambda(p + \varepsilon)$ and $\lambda(p - \varepsilon)$ by~\eqref{eq:lambda_x}.

\State Check the termination condition~\eqref{eq:sprt}.

\EndWhile

\end{algorithmic}
\end{algorithm}

\subsection{Arithmetics of Probabilistic Quantifications}

Now, consider formula~\eqref{eq:smc2}. We denote the satisfaction probability 
of $\P^{\pvs_i} (\varphi_i)$ for each $i \in [n]$ for a given DTMC $\m$ and 
path assignment $V$ by:
\begin{equation} \label{eq:pi}
	\begin{split}
		p_{\varphi_i} = \ & \pr \Big\{ 
		\Big( \pa_l \in \paths \big( V (\pi_l) (0) \big) \Big)_{l \in 
[k_i]} \ :
		\\ & \quad
		\Big( \m, V \big[\pv_l \mapsto \pa_l \textrm{ for all } l \in [k_i] \big] \Big) \models \varphi_i \Big\},
	\end{split}
\end{equation}
% with $\pvs = (\pv_1, \ldots, \pv_n)$ 
where
\[
\quad \pvs_i = (\pv_{i1}, \ldots, \pv_{i k_i}), \quad k_i = \abs{\pvs_i}.	
\]
Again, we assume that each $\varphi_i$ is a bounded-time specification, as we did for \eqref{eq:3a_finite}.
This problem can be converted into the 
(multi-dimensional) HT problem in $\real^n$ by
\begin{equation} \label{eq:ht2}
H_0: \underline{p_\varphi} \in D, \quad H_1: \underline{p_\varphi} \in D^\mathrm{c}, 
\end{equation}
where $D$ is as defined in~\eqref{eq:smc2}, $D^\mathrm{c}$ is the 
complement of $D$, and $\underline{p_\varphi} = (p_{\varphi_1}, \ldots, 
p_{\varphi_n})$. 

We now propose a novel SPRT algorithm for this $n$-dimensional 
HT problem, by extending the common SPRT algorithm from \cref{sub:3a}
to multi-dimension.
By following the same idea, we first generalize the notion of \emph{indifference regions} to the multi-dimensional case.
Based on this, we propose a geometric condition to identify the two most 
indistinguishable cases $\underline{r}$ and $\underline{q}$
from the test regions $D$ and $D^\mathrm{c}$, such that 
it suffices to consider the HT problem:
\begin{equation} \label{eq:simple ht}
H_0': \underline{p_\varphi} = \underline{r}, \quad H_1': \underline{p_\varphi} = \underline{q}.
\end{equation}
Once $\underline{q}$ and $\underline{r}$
in~\eqref{eq:simple ht} are known,
then we can solve it in the same way
as done in \cref{sub:3a}.
Specifically, in~\eqref{eq:smc2},
for each $i \in [n]$, 
we draw $N$ sample path tuples 
for the path variable $\pvs_i$ from the DTMC
and let $T_i$ be the number of sample path tuples,
for which $\varphi_i$ is true.
Consider the log-likelihood function defined~as
\begin{equation} \label{eq:ll}
\lambda(\underline{x}) = \ln \Bigg( \prod_{i \in [n]} x_i^{T_i} (1 - x_i)^{N - T_i} \Bigg),
\end{equation}
where $\underline{x} = (x_1, \ldots, x_n) \in (0,1)^n$.
Clearly, $\lambda(\underline{r})$ and $\lambda(\underline{q})$ 
are the log-likelihood of the two hypotheses 
in~\eqref{eq:simple ht}, respectively. 
So, an SPRT algorithm can be constructed based on the log-likelihood 
ratio $\lambda(\underline{r}) - \lambda(\underline{q})$
(or equivalently $\lambda(\underline{q}) - \lambda(\underline{r})$).
Below, we explain how to derive 
$\underline{q}$ and $\underline{r}$.

\subsection*{Multi-Dimensional Indifference Region}

To ensure that we can find different values for 
$\underline{q}$ and $\underline{r}$  
in~\eqref{eq:simple ht} 
(so that $H_0$ and $H_1$ are not identical),
we introduce a multi-dimensional version 
of the indifference region assumption. 
It ensures that the two 
test regions in~\eqref{eq:ht2} are separated,
as formally stated below. 
This is similar to the case in \cref{sub:3a} 
(see \cref{ass:sprt_indiff} and \cref{rem:indifference}).

\begin{assumption} \label{ass:indifference}
\yw{The test region $D$ is convex}
and there exists convex $D_0, D_1 \subseteq [0,1]^n$, 
such that $D_0 \subseteq D \subset D_1$, 
and the Hausdorff distance $d_{\mathrm{H}} (D_0, D_1) > 0$,
where
$$d_{\mathrm{H}}(X, Y) = \max \Big\{\sup_{x \in X} \inf _{y \in Y} \nm{x - 
y}_2, \ \sup _{y \in Y} \inf _{x \in X} \nm{x - y}_2 \Big\}.$$
For simplicity, we assume that the boundaries of
$D_0$ and $D_1$ are respectively
defined  by 
the boundary equations
\begin{equation}
\label{eq:bounds}
  F_0(\underline{x}) = 0, ~~\text{and~} ~~F_1(\underline{x}) = 0,  
\end{equation}%
where $\underline{x} \in \real^n$, and $F_0$ and $F_1$ are elementary functions.%
\footnote{For example, if the boundary of $D_0$ is a circle 
of radius $0.2$ centered at $(0.5, 0.5)$,
then the elementary function $F_0(p_1, p_2) = (p_1-0.5)^2 + (p_2-0.5)^2 - 0.2^2$.}
\end{assumption}

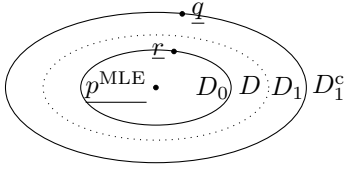
\begin{figure}[!t]
	\centering
	\begin{tikzpicture}
	\node at (0.75, 0) {$D_0$};
	\node at (1.75, 0) {$D_1$};
	\node at (2.3, 0) {$D_1^\mathrm{c}$};
	\node at (1.25, 0.03) {$D$};
	\draw (0, 0) ellipse (1 and .5);
	\draw (0, 0) ellipse (2 and 1);
	\draw[dotted] (0, 0) ellipse (1.5 and .7);
	\draw[fill] (0, 0) ellipse (.03 and .03) node[left] 
{$\underline{\mle}$};
	\draw[fill] (.24, .48) ellipse (.03 and .03) node[left] 
{$\underline{r}$};
	\draw[fill] (.35, .98) ellipse (.03 and .03) node[right] 
{$\underline{q}$};
	\end{tikzpicture}
	\caption{Given the test region $D$, we assume there exists an indifference region formed by $D_1 \backslash D_0$.
	If $\mle$ from \eqref{eq:mle} satisfies
	$\mle \in D_0$, then we find $\underline{r} \in D_0$ 
	by \eqref{eq:r} and $\underline{q} \in D_1^\mathrm{c}$ by \eqref{eq:q}.}
	\label{fig:ht2}
\end{figure}

In general, there exist 
$D_0$ and $D_1$ such that $D_0 \subseteq D \subset D_1$,
when $\underline{p_\varphi}$ is not on the boundary 
of the test region $D$, i.e.,
$\underline{p_\varphi} \notin \partial D$.
Using \cref{ass:indifference}, we derive the HT problem for verifying~\eqref{eq:ht2}
\begin{equation} \label{eq:ht2_sep}
H_0: \underline{p_\varphi} \in D_0, \quad H_1: \underline{p_\varphi} \in D_1^\mathrm{c}, 
\end{equation}
where $\underline{p_\varphi} = (p_{\varphi_1}, \ldots, p_{\varphi_n})$. 
As illustrated in \cref{fig:ht2},
the region $D_1 \backslash D_0$ is the \emph{indifference region}, 
keeping $\underline{p_\varphi}$ statistically distinguishable from 
the boundary of the \emph{test region} $D$.
Again, the HT problem~\eqref{eq:ht2_sep} is \emph{composite}.

\begin{remark}
From \cref{ass:indifference}, if $\overline{D} = \overline{D'}$, 
where $\overline{D}$ and $\overline{D'}$ are respectively 
the closure of $D$ and $D'$, then verifying
$(\P^{\pvs_1} \varphi_1, \allowbreak \ldots, \P^{\pvs_n} \varphi_n) \in D$
is equivalent to verifying $(\P^{\pvs_1} \varphi_1, \ldots, \P^{\pvs_n} 
\varphi_n) \in D'$. Thus, they will not be differentiated in the rest of the 
paper.
\end{remark}

\begin{remark}
\yw{The condition that the test region $D$ 
is convex in \cref{ass:indifference} is only technical.}
If $D$ is non-convex, then we can divide $D$ into 
several convex subregions,
and convert the HT problem~\eqref{eq:ht2}
into several sub-problems with convex test regions. 
From example, the non-convex test region $D$ illustrated in
\cref{fig:nonconex} 
can be divided into the union of two convex
test regions $D_I$ and $D_{\mathit{II}}$.
Therefore, to test if $\underline{p_\varphi} \in D$,
it suffices to test if $\underline{p_\varphi} \in D_I$ or $\underline{p_\varphi} \in D_{II}$, and the overall statistical test error is the sum
of errors of these two sub-tasks.
\end{remark}

\begin{figure}[!t]
	\centering
	\includegraphics[width=1.5in]{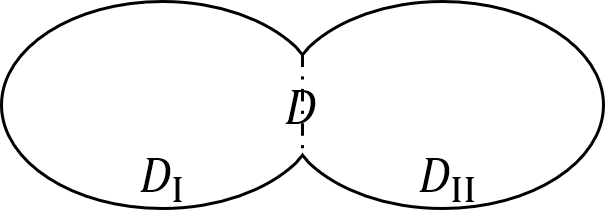}
	\caption{Partition of a non-convex test region.}
	\label{fig:nonconex}
\end{figure}

\subsection*{Identifying Most Indistinguishable Simple Hypotheses}

To solve the HT problem~\eqref{eq:ht2_sep}, 
suppose that we have drawn $N$ sample path tuples
for each path variable $\pvs_i$ ($i \in [n]$).
Let $T_i$ be the number of sample path tuples,
for which $\varphi_i$ is true. Then, we have
\begin{equation} \label{eq:samples}
	T_i \sim \bin(N, p_{\varphi_i}).
\end{equation}
The maximal likelihood estimator (MLE) 
of $\underline{p_\varphi}$ is 
\begin{equation} \label{eq:mle}
	\underline{\mle} = ( \mle_1, \ldots, \mle_n ) = \Big( \frac{T_1}{N}, \ldots, \frac{T_n}{N} \Big).
\end{equation} 
If $\underline{\mle} \in D_0$, then intuitively we should assert the hypothesis $H_0$ against $H_1$.
The statistical error of this assertion
can be measured by the likelihood ratio 
$\lambda(\underline{r}) - \lambda(\underline{q})$
for some $\underline{r} \in D_0$ 
and $\underline{q} \in D_1^\mathrm{c}$,
which will be decided below.
Specifically, to assert $H_0$ (or $H_1$) 
with certain desired FP and FN ratios,
the likelihood ratio should be 
greater (or less) than some threshold,
which is (only) a function of the given FP and FN ratios
(see \cref{sub:app_smc}).

As illustrated in \cref{fig:ht2}, we can identify
$\underline{q} \in D_1^\mathrm{c}$ by maximizing 
the likelihood for the (simple) hypothesis 
$\underline{p_\varphi} = \underline{q}$ 
for any $\underline{q} \in D_1^\mathrm{c}$. 
Intuitively, since any other simple hypothesis 
in $\underline{p_\varphi} \in D_1^\mathrm{c}$ 
yields a larger likelihood ratio,
to use SPRT to solve the HT problem~\eqref{eq:ht2_sep},
it suffices to only consider
the simple hypothesis
$\underline{p_\varphi} = \underline{q}$
from the composite hypothesis
$\underline{p_\varphi} \in D_1^\mathrm{c}$. 
This is formally stated below.

\begin{lemma} \label{lem:q}
If $\underline{\mle} \in D_0$,
to assert $H_0$ (against $H_1$) in the HT problem~\eqref{eq:ht2_sep}, 
it suffices to assert this $H_0$ (against $H_1'$) in the HT problem
\begin{equation} \label{eq:ht3}
H_0: \underline{p_\varphi} \in D_0, \quad H_1': \underline{p_\varphi} = \underline{q}, 
\end{equation}
where $\underline{q}$ is given by
\begin{equation} \label{eq:q}
	\underline{q} = \mathrm{argmax}_{\underline{x} \in D_1^\mathrm{c}} \lambda(\underline{x}),
\end{equation}
with $\lambda(\cdot)$ being the log-likelihood ratio given by~\eqref{eq:ll}.
\end{lemma}

\begin{proof}
Given any possible value of 
$\underline{p_\varphi} \in D_0$,
for any $\underline{q'} \in D_1^\mathrm{c}$,
and any likelihood ratio threshold $B > 0$,
we have 
\[
	\lambda(\underline{p_\varphi}) - \lambda(\underline{q}) > B
	\implies
	\lambda(\underline{p_\varphi}) - \lambda(\underline{q'}) > B,
\]
where $\underline{q}$ is given by~\eqref{eq:q}. 
Thus, for given sample paths (from \eqref{eq:samples}),
if the SPRT algorithm asserts $H_0$ for the HT problem~\eqref{eq:ht3},
then it should also assert $H_0$ for the HT problem~\eqref{eq:ht2_sep}.
The two assertions have the same statistical errors because they use the same 
likelihood ratio threshold~$B$. 
\end{proof}

To obtain $\underline{q}$ from~\eqref{eq:q}, 
by the convexity of 
the test region $D_1$
and the function $\lambda(\cdot)$, 
the maximum is achieved at the boundary of $D_1$. That is, from~\eqref{eq:bounds} it holds that
\begin{equation} \label{eq:q_nc1}
	F_1(\underline{q}) = 0. 
\end{equation}
In addition, by the first-order condition of optimality under the constrained 
\eqref{eq:q_nc1}, the maximum of $\underline{q}$ is achieved when the direction 
of the gradient $\nabla \lambda(\underline{q})$ aligns with the normal vector 
$\nabla F_1(\underline{q})$ of the boundary. That is, for some $c \neq 0$ it holds that
\begin{equation} \label{eq:q_nc2}
\nabla F_1(\underline{q}) = c \nabla \lambda(\underline{q}) = \Big( 
\frac{c (\mle_i - q_i)}{q_i (1 - q_i)} \Big)_{i \in [n]}.
\end{equation}

Given $\underline{q}$ from \cref{lem:q}, 
we identify $\underline{r} \in D_0$ by minimizing 
the Kullback-Leibler divergence from the hypothesis
$\underline{p_\varphi} = \underline{r}$
to the hypothesis $\underline{p_\varphi} = \underline{q}$
for any $\underline{r} \in D_0$, 
as illustrated in \cref{fig:ht2}.
Generally, the Kullback-Leibler divergence measures 
the hardness of using SPRT to distinguish between 
two simple hypotheses
\cite{wald_SequentialTestsStatistical_1945}. 
Thus, to use SPRT to solve the HT problem~\eqref{eq:ht3},
it suffices to only consider the simple hypothesis
$\underline{p_\varphi} = \underline{r}$
from the composite hypothesis
$\underline{p_\varphi} \in D_0$. 
This is formally stated below.

\begin{lemma} \label{lem:r}
If $\underline{\mle} \in D_0$,
to assert $H_0$ (against $H_1'$) in the HT problem~\eqref{eq:ht3}, 
it suffices to assert this $H_0'$ (against $H_1'$) in the HT problem
\begin{equation} \label{eq:ht4}
H_0': \underline{p_\varphi} = \underline{r}, \quad H_1': \underline{p_\varphi} = \underline{q}, 
\end{equation}
where using $\underline{q}$ from \eqref{eq:q}, we have 
\begin{equation} \label{eq:r}
\underline{r} = \mathrm{argmin}_{\underline{x} \in D_0} K ( \underline{x} \Vert \underline{q}),
\end{equation}
where the Kullback-Leibler divergence is given by
\[
K ( \underline{x} \Vert \underline{q}) = 
\sum_{i \in [n]} x_i \ln \Big( \frac{x_i}{q_i} \Big) + (1 - x_i) \ln \Big( \frac{1 - x_i}{1 - q_i} \Big).	
\]
\end{lemma}

\begin{proof}
We defer the proof to \cref{thm:sequential}.
\end{proof}

To solve $\underline{r}$ from~\eqref{eq:r}, by the convexity of 
the test region $D_0$
and the function $K ( \cdot \Vert \underline{q})$, 
the maximum is achieved at the boundary of $D_0$, i.e.,
\begin{equation} \label{eq:r_nc1}
	F_0(\underline{r}) = 0.
\end{equation}
In addition, by the first-order condition of optimality
under the constraint \eqref{eq:r_nc1},
the maximum of $\underline{r}$ is achieved 
when the direction of the gradient 
$\nabla_{\underline{r}} K ( \underline{r} \Vert \underline{q})$ 
aligns with the normal vector $\nabla F_0(\underline{r})$
of the boundary -- i.e., for some $c \neq 0$, it holds~that
\begin{equation} \label{eq:r_nc2}
	\nabla F_0(\underline{r}) =		
	\Big( c \big( \ln(\frac{r_i}{q_i}) - \ln(\frac{1 - r_i}{1 - q_i}) \big) \Big)_{i \in [n]}.
\end{equation}
When $d_{\mathrm{H}} (D_0, D_1) \to 0$, we have $\underline{r} - \underline{q} \to 0$, and~thus
\[
\nabla F_0(\underline{r}) \to \Big( \frac{ c (r_i - q_i) }{q_i (1 - q_i)} \Big)_{i \in [n]}.
\]

The case that $\underline{\mle} \in D_1$ can be handled in the same way.
As shown in \cref{fig:ht1}, we can first derive $\underline{r}$ 
in the same way as \eqref{eq:q} by 
\begin{equation} \label{eq:r'}
	\underline{r} = \mathrm{argmax}_{\underline{x} \in D_0} \lambda(\underline{x}),
\end{equation}
Then, using $\underline{r}$ from \eqref{eq:r'}, we derive $\underline{q}$
in the same way as \eqref{eq:r}; i.e.,
\begin{equation} \label{eq:q'}
	\underline{q} = \mathrm{argmin}_{\underline{x} \in D_1^\mathrm{c}} K ( \underline{x} \Vert \underline{r}).
\end{equation}

\begin{figure}[!t]
	\centering
	\begin{tikzpicture}
	\node at (0.75, 0) {$D_0$};
	\node at (1.75, 0) {$D_1$};
	\node at (2.3, 0) {$D_1^\mathrm{c}$};
	\node at (1.25, 0.03) {$D$};
	\draw (0, 0) ellipse (1 and .5);
	\draw (0, 0) ellipse (2 and 1);
	\draw[dotted] (0, 0) ellipse (1.5 and .7);
	\draw[fill] (0.5, 1.3) ellipse (.03 and .03) node[left] 
{$\underline{\mle}$};
	\draw[fill] (.24, .48) ellipse (.03 and .03) node[left] 
{$\underline{r}$};
	\draw[fill] (.35, .98) ellipse (.03 and .03) node[right] 
{$\underline{q}$};
	\end{tikzpicture}
	\caption{Given the indifference region formed by $D_1 \backslash D_0$, 
	if $\mle$ from \eqref{eq:mle} satisfies
	$\mle \in D_1^\mathrm{c}$, then we find $\underline{r} \in D_0$ 
	by \eqref{eq:r'} and $\underline{q} \in D_1^\mathrm{c}$ by \eqref{eq:q'}.}
	\label{fig:ht1}
\end{figure}
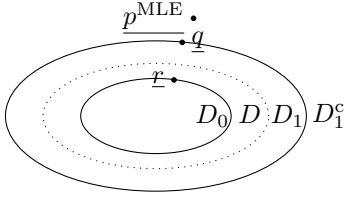

Thus, to achieve the FP and FN ratios $\FP$ and $\FN$, the SPRT algorithm 
should continue sampling until one of the following termination conditions is 
satisfied:
\begin{equation} \label{eq:sprt_comp}
	\begin{cases}
		\text{assert } H_0, 
		\text{ if } \underline{\mle} \in D_0 \text{ and } \lambda(\underline{r}) - \lambda(\underline{q}) > \ln \frac{1 - \FN}{\FP}
		\\ 
		\text{assert } H_1, 
		\text{ if } \underline{\mle} \in D_1^\mathrm{c} \text{ and } \lambda(\underline{q}) - \lambda(\underline{r}) > \ln \frac{1 - \FP}{\FN}
	\end{cases}
\end{equation}
The above discussion is summarized by \cref{thm:sequential} and \cref{alg:non-nested}.

\begin{algorithm}[!t]
\caption{\footnotesize SMC of $(\m, V) \models (\P^{\pvs_1} \varphi_1, \ldots, \P^{\pvs_n} \varphi_n) \in D$.}
\label{alg:non-nested}
\begin{algorithmic}[1]
\Require Desired FP/FN ratios $\FP$/$\FN$, test regions $D_0, D_1$.

\State $N \gets 0$, $T_i \gets 0, \forall i \in [n]$.

\While{True}

\State $N \gets N + 1$.

\For{$i \in [n]$}

\State Draw a tuple of sample paths $\underline{\pa_i}$.

\If{$\varphi_i$ is true on $\underline{\pa_i}$}

\State $T_i \gets T_i + 1$.

\EndIf

\EndFor

\State Compute $\underline{\mle}$ by~\eqref{eq:mle}.

\If{$\underline{\mle} \in D_0$}

\State Compute $\underline{q}$ and $\underline{r}$ by~\eqref{eq:q},~\eqref{eq:r} (via \eqref{eq:q_nc1},~\eqref{eq:q_nc2},~\eqref{eq:r_nc1}, and~ \eqref{eq:r_nc2}).

\ElsIf{$\underline{\mle} \in D_1^\mathrm{c}$}

\State Compute $\underline{r}$ and $\underline{q}$ by~\eqref{eq:r'},~\eqref{eq:q'}.

\Else

\State Continue.

\EndIf

\State Compute $\lambda (\underline{q})$ and $\lambda (\underline{r})$ by \eqref{eq:ll}.

\State Check the termination condition~\eqref{eq:sprt_comp}.

\EndWhile

\end{algorithmic}
\end{algorithm}

\begin{theorem} \label{thm:sequential}
    Under \cref{ass:indifference},
	\cref{alg:non-nested} terminates with probability $1$, and its FP and FN ratios are no greater than $\FP$ and $\FN$.
\end{theorem}

\begin{proof}
Without loss of generality, we consider 
$\underline{p_\varphi} \in D_0$; 
the same applies to the case 
$\underline{p_\varphi} \in D_1^\mathrm{c}$. 
By the central limit theorem,
as the number of samples increases ($N \to \infty$), 
the probability that $\mle \in D_1 \backslash D_0$
converges to $0$.
In addition, the expected value of the log-likelihood ratio 
$\ex (\lambda(\underline{r}) - \lambda(\underline{q})) \to \infty$,
and thus the probability that \eqref{eq:sprt_comp} is not yet satisfied,
converges to $0$.
Therefore, \cref{alg:non-nested} terminates
with probability $1$.

Now, we prove the FP and FN ratios of \cref{alg:non-nested}.
By~\eqref{eq:r}, for any $\underline{r'} \in D_0$, 
the expectation of the log-likelihood ratio 
satisfies
\[
\begin{split}
&
\ex_{\underline{p_\varphi} = \underline{r'}} 
\big( \lambda(\underline{r'}) - \lambda(\underline{q}) \big)
= K ( \underline{r'} \Vert \underline{q}) 
\\ & \qquad 
\geq K ( \underline{r} \Vert \underline{q}) 
= \ex_{\underline{p_\varphi} = \underline{r}} (\lambda(\underline{r}) - \lambda(\underline{q})),
\end{split}
\]
for $K$ from~\eqref{eq:r}.
Therefore, for any $B > 0$, we have 
\[
\pr_{\underline{p_\varphi} = \underline{r'}} (\lambda(\underline{r'}) - \lambda(\underline{q}) > B) \geq \pr_{\underline{p_\varphi} = \underline{r}} (\lambda(\underline{r}) - \lambda(\underline{q}) > B).
\]
This implies that for any possible value of $\underline{p_\varphi} \in D_0$, 
the probability of asserting $H_0$ by~\eqref{eq:sprt_comp} 
using the SPRT is not less than that of $\underline{p_\varphi} = \underline{r}$, which is $1-\FP$.
Therefore, \cref{lem:r} and \cref{thm:sequential} hold. 
\end{proof}

\begin{remark}
For computing~\eqref{eq:q} and \eqref{eq:r}, one can either use optimization or 
solve it via the necessary conditions~\eqref{eq:q_nc1},~\eqref{eq:q_nc2},~\eqref{eq:r_nc1}, and \eqref{eq:r_nc2}, which may have analytic solutions as 
the boundary functions $F_0 (\cdot)$ and $F_1 (\cdot)$ are elementary functions 
(especially when $F_0 (\cdot)$ and $F_1 (\cdot)$ are linear functions). 
Since solving the optimization problem at every iteration can be inefficient for some cases this, we can reduce the frequency of computing the significance level by drawing samples in batches.
\end{remark}

\subsection{Nested probabilistic quantification}

The nested probabilistic quantification in \hpctls can be handled in the same way as~\cite{wzbp19}.
Thus, the nested probability operators in~\eqref{eq:smc3} can be handled in the 
same way as done in~\cite{wzbp19}, and we omit describing it here.

%\vspace{-5pt}
\section{Case Studies and Evaluation}
\label{sec:simulation}
%\vspace{-5pt}

We evaluated the presented SMC algorithms on the case studies
described in \cref{sec:applications}. It is important to highlight that all 
these \hpctls specifications are currently not verifiable by existing 
probabilistic model checkers and SMC tools.
The simulations were performed on a laptop with Intel\textregistered\
Core\texttrademark\ i7-7820HQ, 2.92GHz Processor with 32GB RAM.
The simulation code is available at~\cite{gitlab_hpctls}. 
The assertions of the proposed SMC algorithms are compared with ``the correct
answers'', which are derived by extensive simulations or exhaustive solutions.
These ``the correct answers'' are also used to check 
the validity of the indifference region assumption on the case studies.
The running time, number of samples, and the accuracy of the proposed
algorithms (Number of correct assertions / Number of total assertions) are
estimated based on $100$ runs for each SMC task.
% approximate equivalence parameter $\varepsilon$ and model sizes $N$.
The results are presented in \cref{tb:sc,tb:pn,tb:dc,tb:cache}, respectively.
In all the setups, the estimated accuracy agrees with the fixed desired 
significance levels ($\FP = \FN = 0.01$), except for one case in \cref{tb:dc}.
This is because of the statistical error of the estimated accuracy using 
only $100$ runs. 
The average execution time in the worst case is less than~$30$~seconds.

\subsection{Side-channel Vulnerability}
We verified the correctness of the \hpctls specification~\eqref{eq:sc_vul} on 
GabFeed chat server~\cite{GabFeed}. The authentication algorithm in this version 
of GabFeed has been reported to have a side channel vulnerability that leaks the 
number of set bits in the secret key~\cite{tizpaz2018data}. 
The vulnerability can be exploited by the attacker by observing the execution time across different public keys, as discussed in \cref{sec:applications}; hence, as with~\cite{tizpaz2018data},
we verify the security policy~\eqref{eq:sc_vul} for a selection of security keys.
We instrumented the source code to obtain the execution time for a combination of the secret key and public key, and generate a trace in a discrete-time fashion. 
For a given secret key, we select a random public key and generate a trace from 
it. Using this approach we were able to show the existence of side-channel -- i.e., the negation of~\eqref{eq:sc_vul} holds with confidence level 
$0.99$.
The results are shown in \cref{tb:sc}.

\setlength{\tabcolsep}{12pt}

\begin{table*}[h]
\centering
%\footnotesize
%\scalebox{0.8}
{
\begin{tabular}{rrrrrr}
\toprule
   $\tau$ &  $\epsilon$& $\delta$ & Acc. & No. Samples & Time (s) \\
\midrule

 60  &     0.05  &     0.01  &     1.00 &     5.5e+02 &     0.54 \\
 60  &     0.05  &     0.001 &     1.00 &     5.5e+03 &     5.76 \\
 60  &     0.1   &     0.01  &     1.00 &     6.1e+02 &     0.60 \\
 60  &     0.1   &     0.001 &     1.00 &     6.2e+03 &     7.16 \\
 90  &     0.05  &     0.01  &     1.00 &     3.7e+02 &     0.46 \\
 90  &     0.05  &     0.001 &     1.00 &     3.7e+03 &     4.94 \\
 90  &     0.1   &     0.01  &     1.00 &     4.1e+02 &     0.48 \\
 90  &     0.1   &     0.001 &     1.00 &     4.1e+03 &     5.37 \\
 120 &     0.05  &     0.01  &     1.00 &     3.8e+02 &     6.96 \\
 120 &     0.05  &     0.001 &     1.00 &     2.2e+03 &     11.24 \\
 120 &     0.1   &     0.01  &     1.00 &     3.8e+02 &     6.05 \\
 120 &     0.1   &     0.001 &     1.00 &     2.3e+03 &     9.46 \\

\bottomrule
\end{tabular}
}
\caption{Showing the violation of timing side channel vulnerability for 
different combinations of time thresholds $\tau$ seconds, approximate 
equivalence parameter $\epsilon$ and indifference region $\delta$ based on the 
average of $100$ runs.\label{tb:sc}}
%\vspace{-10pt}
\end{table*}

% \vspace{-4pt}
% \paragraph{Probabilistic causation} We verified correctness of the \hpctls specification~\eqref{eq:def_pc} with a bounded time horizon $100$ on randomly generated labeled DTMCs with $N \in \{50, 100, 300, 1000\}$ states for significance level $\alpha = 0.05$, as shown in \cref{tb:pc}.
% As the specification is non-nested, the sample costs for different model sizes are similar.
% The exhibited increase in the running time is mainly resulted from the growing computational cost in drawing sample paths from larger models.

\subsection{Probabilistic Noninterference} We showed
the violation of specification~\eqref{eq:def_pn}
for $N \in \{20, 50, 100\}$ threads (the results are similar for $l=1$). The obtained results are presented in \cref{tb:pn}.
The total number of states of the DTMC is at least $N !$, so we simulate it using a transition-matrix-free approach to meet the memory constraint.
As the significance level decreases, namely a more accurate assertion is asked for, the sample cost and the running time increase accordingly.

\begin{table*}[h]
\centering
%\footnotesize
%\scalebox{0.8}
{
\begin{tabular}{rrrrr}
\toprule
   $N$ &  $\delta$ & Acc. & No. Samples & Time (s) \\
\midrule
20  & 0.01   &    1.00 &    7.7e+02 &    0.49 \\
20  & 0.001  &    1.00 &    7.6e+03 &    6.45 \\
50  & 0.01   &    1.00 &    7.0e+02 &    0.48 \\
50  & 0.001  &    1.00 &    6.8e+03 &    6.39 \\
100 & 0.01   &    1.00 &    6.5e+02 &    0.54 \\
100 & 0.001  &    1.00 &    6.6e+03 &    7.10 \\
\bottomrule
\end{tabular}
}
\caption{Showing the violation of probabilistic noninterference for different 
combinations of number of threads $N$ and  indifference region $\delta$, based 
on the average of~$100$~runs.\label{tb:pn}}
% \vspace{-10pt}
\end{table*}

\subsection{Security of Dining Cryptographers} We verified the correctness of 
the
specification~\eqref{eq:def_dc} with $i = 1, j = 2$ on the model provided
by~\cite{_PRISMCaseStudies_} for $N \in \{100, 1000\}$ cryptographers
and approximate equivalence parameter $\epsilon \in \{0.2, 0.1, 0.0.5\}$. The obtained results are  summarized in \cref{tb:dc}.
The total number of states of the DTMC is at least $2^N$, and we simulate it with a transition-matrix-free approach.
As the approximate equivalence parameters increases, the specification is increasingly relaxed, so the sample cost and the running time decrease accordingly.

\begin{table*}[h]
\centering
%\footnotesize
%\scalebox{0.8}
{
\begin{tabular}{rrrrrr}
\toprule
   $N$ & $\epsilon$ & Acc. & No. Samples & Time (s) \\
\midrule
 100  &    0.05 &    1.00 &    1.0e+03 &     0.91 \\
 100  &    0.1  &    1.00 &    5.2e+02 &     0.39 \\
 100  &    0.2  &    1.00 &    2.8e+02 &     0.14 \\
 1000 &    0.05 &    \textbf{0.98} &    1.1e+03 &     3.27 \\
 1000 &    0.1  &    1.00 &    5.5e+02 &     1.52 \\
 1000 &    0.2  &    1.00 &    2.8e+02 &     0.69 \\
\bottomrule
\end{tabular}
}
\caption{Verifying the security of dining cryptographers for different 
combinations of number of cryptographers $N$ and approximate equivalence 
parameter $\epsilon$ for indifference region $\delta = 0.01$, based on the 
average of $100$ runs.\label{tb:dc}}
% \vspace{-6pt}
\end{table*}

\subsection{Randomized cache Replacement Policy}
We verified the correctness of the specification~\eqref{eq:def_rcr} 
for the performance of random replacement cache policy described in \cref{sec:applications}.
The performance of random replacement policy is evaluated on random memory accesses from a normal distribution with variance less than the cache size, to emulate the \emph{locality of reference}.
With the random replacement policy and the random access sequence,
the dynamics of the cache modeled by the Mealy machine described in \cref{sec:applications} can be captured by a DTMC.

We consider the paths of the DTMC with labels $\mathtt{H}$ or $\mathtt{M}$, depending on the outcome of the cache access. 
We compared the probability of all hits 
to the probability of seeing a single miss $\mathtt{M}$
on a fully associative cache with $256$ lines 
for a program of $1024$ blocks. 
This can easily be extended to set associative cache with arbitrary program size. 
The results are shown in \cref{tb:cache}.  
We observe that the algorithm takes longer time 
for $T=20$ than $T=10$. 
This is because, for shorter $T$, the probability of observing
all hits $\mathtt{H}$ is more than the probability of observing a miss $\mathtt{M}$. 
As the trace length increases, these probabilities become closer.
%

% \begin{table*}[!t]
% \centering
% \footnotesize
% \begin{tabular}{rrrrll}
% \toprule
%   $N$ &  Acc. & No. Samples & Time (s) \\
% \midrule
%   50  &   1.00 &  1.5e+04 &  1.83 \\
%  100  &   0.98 &  3.1e+04 &  3.47 \\
%  300  &   1.00 &  5.0e+04 &  10.25 \\
%  1000 &   0.99 &  6.9e+04 &  17.95 \\
% \bottomrule
% \end{tabular}
% \caption{ Estimated time, number of samples (No. Samples) and accuracy (Acc.) for verifying probabilistic causation for different number of states $N$ for significance level $\alpha = 0.05$, based on the average of $100$.\label{tb:pc}}
% \vspace{-10pt}
% \end{table*}

\begin{table*}[!t]
\centering
%\footnotesize
%\scalebox{0.8}
{
\begin{tabular}{rrrrrr}
\toprule
$T$ & $\epsilon$ & $\delta$ & Acc. & No. Samples & Time (s) \\
\midrule
10 & 0.05 & 0.01 & 1.00 & 1.1e+02 & 0.13 \\
10 & 0.05 & 0.001 & 1.00 & 1.0e+03 & 2.56 \\
10 & 0.01 & 0.01 & 1.00 & 1.2e+02 & 0.14 \\
10 & 0.01 & 0.001 & 1.00 & 1.2e+03 & 2.79 \\
20 & 0.05 & 0.01 & 1.00 & 6.0e+02 & 1.49 \\
20 & 0.05 & 0.001 & 1.00 & 6.2e+03 & 16.73 \\
20 & 0.01 & 0.01 & 0.99 & 1.2e+03 & 2.97 \\
20 & 0.01 & 0.001 & 1.00 & 1.1e+04 & 28.99 \\
\bottomrule
\end{tabular}
}
\caption{Verifying the performance of random replacement policy for different combinations of trace length ($T$) and approximation parameter ($\epsilon$) and indifference region ($\delta$), based on the average of $100$ runs.\label{tb:cache}}
% \vspace{-6pt}
\end{table*}
\section{Related Work}
\label{sec:related}

To the best of our knowledge, the only existing SMC algorithm for hyper temporal logics is the one proposed in~\cite{wzbp19}.
It handles complex probabilistic quantifications
similar to \hpctls but using a multi-dimensional extension of
Clopper-Pearson confidence interval, whereas, in this paper, our focus is on 
SPRT. Moreover, the application domain of~\cite{wzbp19} is on timed 
hyperproperties and cyber-physical systems, whereas, here, we 
concentrate on applications in information-flow security. 
This algorithm provides provable probabilistic guarantees
for any desired false positive $\FP \in (0,1)$
(the probability of wrongly claiming a false formula to be true)
and false negative $\FN \in (0,1)$
(the probability of wrongly claiming a true formula to be false).

Randomization is used in different contexts to quantify the amount 
of information leak as well as to provide probabilistic guarantees about the correctness of security policies. A classic example is probabilistic 
noninterference~\cite{gray90,g92}, which requires that high-security input 
should not change the probability of reaching low-security
% \todo{equivalent? you wanted to say security} 
outputs. There has been extensive work in this area including using probabilistic bisimulation to reason about probabilistic noninterference in multi-threaded programs~\cite{ss00}. Another prominent line of work is {\em quantitative 
information flow}~\cite{s09,kb07}, which relates information theory to 
independent executions of a system and uses different notions of entropy to 
quantify the amount information leaked across different executions. 

Recently, there has been significant progress in automatically
{verifying}~\cite{frs15,fmsz17,fht18,cfst19}
and {monitoring}~\cite{ab16,fhst19,bsb17,bss18,fhst18,sssb19,hst19} 
\hltl specifications. \hltl is also supported by a growing set of 
tools, including the model checker MCHyper~\cite{frs15,cfst19}, the  
satisfiability checkers EAHyper~\cite{fhs17} and MGHyper~\cite{fhh18}, and the runtime monitoring tool RVHyper~\cite{fhst18}. Synthesis techniques for 
\hltl are studied in~\cite{fhlst18} and in~\cite{bf19}.

%\vspace{-5pt}
\section{Conclusion} \label{sec:conclusion}
%\vspace{-5pt}

In this paper, we studied the problem of statistical model checking (SMC) 
of hyperproperties on discrete-time Markov chains (DTMCs).
First, to reason about probabilistic hyperproperties, we introduced the
probabilistic temporal logic \hpctls that extends \pctls by allowing explicit
and simultaneous quantification over paths. 
%
% \hpctls also generalizes \hpctl by
% incorporating nested temporal and probability operators.
%
In addition, we proposed an SMC algorithm for \hpctls specifications on DTMCs.
Unlike existing SMC algorithms for hyperproperties based on Clopper-Pearson 
confidence interval, we proposed sequential probability ratio tests (SPRT) 
with a new notion of indifference margin. 
Finally, we evaluated our SMC algorithms on four case studies:
time side-channel vulnerability in encryption, 
probabilistic anonymity in dining cryptographers, 
probabilistic noninterference of parallel programs,
and the performance of a random cache replacement policy.

For future work, we are currently developing SMC algorithms for verification of 
timed hyperproperties in probabilistic systems. Another interesting research 
avenue is developing exhaustive model checking
algorithms for \hpctls. One can also develop symbolic techniques for
verification of \hpctls specifications. We also note that our approach has the 
potential of being generalized to reason about the conformance of two systems 
(e.g., an abstract model and its refinement) with respect to hyperproperties.

\section*{Acknowledgment}
\addcontentsline{toc}{section}{Acknowledgment}
This work is sponsored in part by the ONR under agreements 
N00014-17-1-2504 and N00014-20-1-2745, AFOSR under award number 
FA9550-19-1-0169, as well as the NSF CNS-1652544 and NSF SaTC-1813388 grant.

\bibliographystyle{IEEEtranS}
\bibliography{ref}

% \newpage

\appendix

We recap the temporal logics 
relevant to this paper with the notations adapted to that of \hpctls,
and introduce the basics on the statistical model checking of \pctls
using the sequential probability ratio test (SPRT).

% \todo{BB: Revise}

\subsection{PCTL$^*$} \label{sub:app_pctls}

\paragraph{Syntax} The syntax 
of \pctls~\cite{baier_PrinciplesModelChecking_2008} consists of \emph{state} 
formulas $\Phi$ and \emph{path} formulas $\varphi$ that are defined respectively 
over the set of atomic propositions $\aps$ by: 
\[
  \Phi \Coloneqq \ \ap 
  \ \vert \ \neg \Phi
  \ \vert \ \Phi \land \Phi
  \ \vert \ \P^{J} (\varphi)
\]
and
\[
  \varphi \Coloneqq \ \Phi 
  \ \vert \ \neg \varphi
  \ \vert \ \varphi \land \varphi
  \ \vert \ \X \varphi
  \ \vert \ \varphi \U^{\leq k} \varphi
\]
where $\ap \in \aps$, and $J \subseteq [0, 1]$ is an interval with rational~bounds.

\paragraph{Semantics}
The satisfaction relation $\models$ of the \pctls state and path formulas is defined for a state and a path of a labeled DTMC $\m$ respectively by 
\[
\begin{array}{l@{\hspace{1em}}c@{\hspace{1em}}l}
(\m, \ms) \models \ap & \textrm{iff} & \ap \in \ml (\ms)
\\
(\m, \ms) \models \neg \Phi & \textrm{iff} & (\m, \ms) \not\models \Phi
\\
(\m, \ms) \models \Phi_1 \land \Phi_2 &  \textrm{iff} & (\m, \ms) \models \Phi_1 \textrm{ and } (\m, \ms) \models \Phi_2
\\
(\m, \ms) \models \P^{J} (\varphi) & \textrm{iff} 
&
\pr \big( (\m, \ms) \models \varphi \big) \in J 
\end{array}
\]
and
\[
\begin{array}{l@{\hspace{0.5em}}c@{\hspace{0.5em}}l}
(\m, \pa) \models \Phi & \textrm{iff} & (\m, \pa(0)) \models \Phi
\\
(\m, \pa) \models \neg \varphi & \textrm{iff} & (\m, \pa) \not\models \varphi
\\
(\m, \pa) \models \varphi_1 \land \varphi_2 &  \textrm{iff} & (\m, \pa) \models \varphi_1 \textrm{ and } (\m, \pa) \models \varphi_2
\\
(\m, \pa) \models \X \varphi & \textrm{iff} & (\m, \pa^{(1)}) \models \varphi
\\
(\m, \pa) \models \varphi_1  \U^{\leq k}  \varphi_2  & \textrm{iff} &
\textrm{there exists } i \leq k \textrm{ such that } 
 \\ & & 
 \big((\m, \pa^{(i)}) \models \varphi_2 \big) \land
 \\ & & 
 \big(\text{for all } j < i, (\m, \pa^{(j)}) \models 
\varphi_1 \big)
\end{array}
\]
where $\pa^{(i)}$ is the $i$-suffix of path $\pa$.

\subsection{HyperLTL}

\paragraph{Syntax} HyperLTL~\cite{clarkson_TemporalLogicsHyperproperties_2014} 
formulas are defined over the set of atomic propositions $\aps$ respectively by:
\[
  \psi \Coloneqq \ 
  \exists \pv. \ \psi
  \ \vert \ \forall \pv. \ \psi
  \ \vert \ \varphi
\]
and
\[
  \varphi \Coloneqq \ 
  \ap^\pv 
  \ \vert \ \neg \varphi
  \ \vert \ \varphi \land \varphi
  \ \vert \ \X \varphi
  \ \vert \ \varphi \U \varphi
\]
where $\ap \in \aps$.

\paragraph{Semantics} The semantics of 
\hltl is defined for a trace assignment $V: \Pi \to (\nat \to 2^\aps)$ by:
\[
  \begin{array}{l@{\hspace{1em}}c@{\hspace{1.2em}}l}
    V \models \ap^\pv & \textrm{iff} & \ap \in \ml (V(\pi(0)))
    \\
    V \models \exists \pv. \ \psi & \textrm{iff} & \text{there
    exists } \pa \in T 
    \\ & &
    \text{ such that } V[\pv \mapsto \pa] \models \psi
    \\
    V \models \forall \pv. \ \psi & \textrm{iff} & \text{for all } \pa \in T 
    \\ & & 
    \text{ such that } V[\pv \mapsto \pa] \models \psi
    \\
    V \models \neg \varphi & \textrm{iff} & V \not\models \varphi
    \\
    V \models \varphi_1 \land \varphi_2 &  \textrm{iff} & V \models 
\varphi_1 \textrm{ and } V \models \varphi_2
    \\
    V \models \X \varphi & \textrm{iff} & V^{(1)} \models \varphi
    \\
    V \models \varphi_1  \U  \varphi_2  & \textrm{iff} &
    \textrm{there exists } i\geq 0 \textrm{ such that } 
    \\ & & 
    \big(T,V^{(i)} \models \varphi_2 \big) \land 
    \\ & & 
    \big(\text{for all } j < i, \text{ we have } V^{(j)} \models 
\varphi_1 \big)
  \end{array}
\]
where $V^{(i)}$ is the $i$-shift of path assignment $V$, defined by 
$V^{(i)}(\pv) = (V(\pv))^{(i)}$.

\subsection{HyperPCTL}

\paragraph{Syntax} \hpctl~\cite{ab18} formulas 
are defined over the set of atomic propositions $\aps$ respectively by: 
\[
  \psi \Coloneqq \ 
  \ap_\sv
  \ \vert \ \exists \sv. \ \psi
  \ \vert \ \forall \sv. \ \psi
  \ \vert \ \neg \psi
  \ \vert \ \psi \land \psi 
  \ \vert \ p \Join p
\]
\[
  p \Coloneqq \ 
  \P (\varphi) 
  \ \vert \ c
  \ \vert \ p + p
  \ \vert \ p - p
  \ \vert \ p \cdot p
\]
\[
  \varphi \Coloneqq \ 
  \X \psi
  \ \vert \ \psi \U^{\leq k} \psi
\]
where $\ap \in \aps$, $c \in \mathbb{Q}$ and $\Join \in \{<,>,\leq,\geq,=\}$.

\paragraph{Semantics} The satisfaction relation $\models$ of \hpctl is defined 
for state and path formulas of a labeled DTMC $\m$ respectively by: 
\[
\begin{array}{l@{\hspace{1em}}c@{\hspace{1.2em}}l}
(\m, X) \models \ap_\sv & \textrm{iff} & \ap \in X (\sv)
\\
(\m, X) \models \exists \sv. \ \psi & \textrm{iff} & \text{there
exists } \ms \in \mss 
\\ & & 
\text{ such that } X[\sv \mapsto \ms] \models \psi
\\
(\m, X) \models \forall \sv. \ \psi & \textrm{iff} & \text{for all } \ms \in 
\mss 
\\ & &
\text{ such that } X[\sv \mapsto \ms] \models \psi
\\
(\m, X) \models \neg \psi & \textrm{iff} & (\m, X) \not\models \psi
\\
(\m, X) \models \psi_1 \land \psi_2 &  \textrm{iff} & (\m, X) \models \psi_1 
\\ & &
\textrm{ and } (\m, X) \models \psi_2
\\
(\m, X) \models p_1 \Join p_2 & \textrm{iff} & \llbracket p_1 \rrbracket_{(\m, X)} \Join \llbracket p_1 \rrbracket_{(\m, X)}
\end{array}
\]

\[
\begin{array}{l@{\hspace{0.5em}}c@{\hspace{0.8em}}l}
  \llbracket c \rrbracket_{(\m, X)} & = & c   
  \\
  \llbracket p_1 + p_2 \rrbracket_{(\m, X)} & = & \llbracket p_1 \rrbracket_{(\m, X)} + \llbracket p_2 \rrbracket_{(\m, X)}
  \\
  \llbracket p_1 - p_2 \rrbracket_{(\m, X)} & = & \llbracket p_1 \rrbracket_{(\m, X)} - \llbracket p_2 \rrbracket_{(\m, X)}  
  \\
  \llbracket p_1 \cdot p_2 \rrbracket_{(\m, X)} & = & \llbracket p_1 \rrbracket_{(\m, X)} \cdot \llbracket p_2 \rrbracket_{(\m, X)}  
  \\
  \llbracket \P (\varphi) \rrbracket_{(\m, X)}
  & = & \pr \big\{ (\pa_i \in \paths(X(\pv_i))_{i \in [n]} \mid 
  \\ & & \quad (M, \underline{\pa}) \models \varphi \big\} 
\end{array}
\]

\[
  \begin{array}{l@{\hspace{0.5em}}c@{\hspace{0.8em}}l}
    (M, \underline{\pa}) \models \X \varphi & \textrm{iff} & (M, \underline{\pa}(1)) \models \varphi
    \\
    (M, \underline{\pa}) \models \varphi_1  \U^{\leq k}  \varphi_2  & \textrm{iff} &
    \textrm{there exists } i \leq k \textrm{ such that } 
    \\ & & \big((M, 
    \underline{\pa}(i)) \models \varphi_2 \big) \land 
    \\ & & \big(\text{for all } j < i, (M, \underline{\pa}(i)) 
\models \varphi_1 \big)
  \end{array}
\]
where $\underline{\pa} = (\pa_1, \ldots, \pa_n)$ and $\underline{\pa}(i) = (\pa_1(i), \ldots, \pa_n(i))$.

\subsection{Statistical Model Checking of \pctls} \label{sub:app_smc}

The key issue in the statistical model checking of \pctls
is to deal with the probabilistic operator by sampling
\cite{larsen_StatisticalModelChecking_2016,agha_SurveyStatisticalModel_2018}.
Specifically, consider the satisfaction of
a \pctls formula
\begin{equation} \label{eq:app_smc}
	(\m, \ms) \models \P^{[0, p]} (\varphi)
\end{equation}
where $\varphi$ is a
linear temporal logic formula
and $p \in (0,1)$ is a given real number.
From the semantics of \pctls in \cref{sub:app_pctls},
it means that the satisfaction probability 
\begin{equation} \label{eq:app_pphi}
	p_\varphi = \pr \big( (\m, \ms) \models \varphi \big)
\end{equation}
of $\varphi$ for a given model $\m$
with the initial state $\ms$ satisfies
$$p_\varphi \leq p$$
For simplicity, we assume $\varphi$
is bounded-time and contains no 
probabilistic operator, thus its truth value
can be decided on
finite-length sample paths of $\m$.
Unbounded-time non-nested formulas can be handled
similarly with extra considerations.

To statistically infer \eqref{eq:app_smc}, the assumption of 
an \emph{indifference region} is usually adopted
\cite{larsen_StatisticalModelChecking_2016,agha_SurveyStatisticalModel_2018}.
That is, there exists $\varepsilon > 0$, such that 
the satisfaction probability $p_\varphi$ from \eqref{eq:app_pphi} satisfies
\begin{equation} \label{eq:app_indifference}
	p_\varphi \notin (p - \varepsilon, p + \varepsilon),
\end{equation}
where $(p - \varepsilon, p + \varepsilon) \subseteq [0, 1]$.
The interval $(p - \varepsilon, p + \varepsilon)$
is commonly referred to as the {indifference region}.
By the assumption \eqref{eq:app_indifference},
to statistically model check~\eqref{eq:app_smc},
it suffices to consider the hypothesis testing (HT) problem
\begin{equation} \label{eq:app_ht}
\begin{split}
	& H_0: p_\varphi \leq p - \varepsilon, 
	\\ & H_1: p_\varphi \geq p + \varepsilon.
\end{split}
\end{equation}
and infer whether $H_0$ or $H_1$ holds by sampling.

The HT problem \eqref{eq:app_ht} is \emph{composite},
since it contains (infinitely) many \emph{simple}
HT problems:
\begin{equation} \label{eq:app_simples}
	\begin{split}
		& H_0^{p_0}: p_\varphi = p_0, 
		\\ & H_1^{p_1}: p_\varphi = p_1.
	\end{split}
\end{equation}
where $p_0$ and $p_1$ can take values from 
$[0, p - \varepsilon]$ and 
$[p + \varepsilon, 1]$, respectively.
Intuitively, among all the possible 
values of $p_0$ and $p_1$,
the following is the most ``indistinguishable'':
\begin{equation} \label{eq:app_simp}
	\begin{split}
		& H_0^{p - \varepsilon}: p_\varphi = p - \varepsilon, 
		\\ & H_1^{p + \varepsilon}: p_\varphi = p + \varepsilon.
	\end{split}
\end{equation}
(We will discuss the meaning of ``indistinguishable'' later.)

To solve the hypothesis testing problem~\eqref{eq:app_simp},
we draw statistically independent sample paths  
$\pa_1, \pa_2, \ldots$
from the given model $\m$
with the initial state $\ms$.
For $N$ such samples, 
the log-likelihood of observing $N$
such sample paths under the two hypotheses 
$H_0^{p - \varepsilon}$ and $H_1^{p + \varepsilon}$,
are respectively $\lambda(p - \varepsilon)$ and $\lambda(p + \varepsilon)$,
where 
\begin{equation} \label{eq:app_lambda}
	\lambda_{N,T} (p) = \ln \big( p^{T} (1 - p)^{N - T} \big).	
\end{equation}
Accordingly, the log-likelihood ratio of the two hypotheses is 
\begin{equation} \label{eq:app_llr}
	\Lambda_{N,T} (p + \varepsilon, p - \varepsilon) = \lambda(p  +\varepsilon) - \lambda(p - \varepsilon).	
\end{equation}
Clearly, as $\Lambda_{N,T} (p + \varepsilon, p - \varepsilon)$ increases, 
$H_1^{p + \varepsilon}$ is more likely to be true, and 
the less statistical error is made 
when we assert $H_1^{p + \varepsilon}$ is true;
and \emph{vice versa}.

The sequential probability ratio test (SPRT)
explicitly tells us how to make these assertions from $\Lambda_{N,T}$
to achieve certain levels of statistical errors 
\cite{wald_SequentialTestsStatistical_1945}.
The statistical errors are formally given by
the probability of falsely asserting $H_1^{p + \varepsilon}$ while $H_0^{p - \varepsilon}$ holds, and the probability of falsely asserting $H_0^{p - \varepsilon}$ while $H_1^{p + \varepsilon}$ holds:
\[
\begin{split}
	& \FP = \pr \big( \text{SPRT assert } H_1^{p + \varepsilon} \mid  H_0^{p - \varepsilon} \textrm{ is true} 
\big), 
	\\ & \FN = \pr \big( \text{SPRT assert } H_0^{p - \varepsilon} \mid  H_1^{p + \varepsilon} \textrm{ is true} 
\big),
\end{split}
\]
where $\FP$ and $\FN$ are called \emph{false positive} (FP) and \emph{false negative} (FN) ratios.
When $\FP = \FN$, we may also refer to them as the \emph{significance level}.

To achieve the given desired $\FP$ and $\FN$, the SPRT is implemented \emph{sequential}
-- i.e., it continuously draws samples until the 
following condition is satisfied:
\begin{equation} \label{eq:app_sprt}
	\begin{cases}
		\text{assert } H_0^{p - \varepsilon}, & \text{ if } \Lambda_{N,T} (p + \varepsilon, p - \varepsilon) < \ln \frac{\FP}{1 - \FN} \\
		\text{assert } H_1^{p + \varepsilon}, & \text{ if } \Lambda_{N,T} (p + \varepsilon, p - \varepsilon) > \ln \frac{1 - \FP}{\FN} 
	\end{cases}
\end{equation}
It can be proved that this SPRT algorithm always terminates with probability $1$
and it strictly achieves the desired $\FP$ and $\FN$
\cite{wald_SequentialTestsStatistical_1945}.

Finally, we discuss the ``indistinguishability''.
Suppose the true satisfaction probability satisfies 
$p_\varphi \in [p + \varepsilon, 1]$.
(The case $p_\varphi \in [0, p - \varepsilon]$ is similar.)
Following \eqref{eq:app_sprt}, if the SPRT asserts $H_1^{p + \varepsilon}$
for certain $N$ and $T$,
then we have that 
\[
\Lambda_{N,T} (p_\varphi, p - \varepsilon) > \ln \frac{1 - \FP}{\FN}.
\]
This means that if $N$ and $T$ are sufficient to assert $H_1^{p + \varepsilon}$
against $H_0^{p - \varepsilon}$ with the desired $\FP$ and $\FN$, then 
they are sufficient to assert $H_1^{p_\varphi}$
against $H_0^{p - \varepsilon}$.
In other words, $p + \varepsilon$ is the worst case.
Similarly, if the SPRT asserts $H_0^{p - \varepsilon}$, 
then we can show that 
\[
\Lambda_{N,T} (p_\varphi, p - \varepsilon) < \ln \frac{\FP}{1 - \FN},
\]
and the same argument follows.
For more detailed discussions, please refer to~\cite{sen_StatisticalModelChecking_2004}.

\end{document}